\documentclass[12pt]{article}
\usepackage{amsmath}
\usepackage{graphicx}
\usepackage{enumerate}
\usepackage{natbib}
\usepackage{url} 
\usepackage{lscape} 

\usepackage[utf8]{inputenc}
\usepackage{amsthm}
\usepackage{amssymb}
\usepackage{xcolor}
\usepackage{comment}
\usepackage{appendix}
\usepackage{bm}
\usepackage{breakcites}

\definecolor{darkRed}{rgb}{0.60,.03,.03}

\usepackage{tikz}
\usetikzlibrary{shapes,snakes, graphs}

\usetikzlibrary{arrows,%
                petri,%
                topaths}%
\usepackage{pgf}
\usepackage{subfigure,tikz}
\usetikzlibrary{arrows,automata}
\usepackage{svg}

\usepackage[backref=page]{hyperref}
\hypersetup{
    colorlinks=true,
    linkcolor=blue,
    filecolor=magenta,
    citecolor=darkRed,
    urlcolor=cyan,
    pdftitle={On the Causal Interpretation of Randomized Interventional Indirect Effects},
    }

\newcommand{\blind}{1}

\newtheorem{assumption}{Assumption}
\newtheorem{theorem}{Theorem}

\newtheorem{corollary}{Corollary}

\newtheorem{definition}{Definition}
\def\ci{\mbox{\ensuremath{\perp\!\!\!\perp}}}

\allowdisplaybreaks[1]

\begin{document}

\def\spacingset#1{\renewcommand{\baselinestretch}%
{#1}\small\normalsize} \spacingset{1}


\if1\blind { \title{On the Causal Interpretation of Randomized Interventional Indirect Effects} \author{Caleb  H.~Miles \thanks{Caleb
      H.  Miles is Assistant  Professor, Department of Biostatistics, Columbia
      University, New  York, NY 10032,  USA.}\hspace{.2cm}} 
  \date{}
    \maketitle
} \fi

\if0\blind
{
  \bigskip
  \bigskip
  \bigskip
  \begin{center}
    {\LARGE\bf On the Causal Interpretation of Randomized Interventional Indirect Effects}
\end{center}
  \medskip
} \fi

\bigskip

\begin{abstract}
  \noindent Identification of standard mediated effects such as the natural indirect effect relies on heavy causal assumptions. By circumventing such assumptions, so-called randomized interventional indirect effects have gained popularity in the mediation literature. Here, I introduce properties one might demand of an indirect effect measure in order for it to have a true mediational interpretation. For instance, the sharp null criterion requires an indirect effect measure to be null whenever no individual-level indirect effect exists. I show that without stronger assumptions, randomized interventional indirect effects do not satisfy such criteria. I additionally discuss alternative causal interpretations of such effects.

\end{abstract}
\noindent%
  \vfill

\newpage
\spacingset{1.5} 

\section{Introduction}
\label{sec:intro}
Mediation analysis is a widely popular discipline whose underlying objective is to understand the mechanism by which an exposure (or treatment) affects an outcome, in particular by studying intermediate variables that might be responsible for transmitting such an effect. It is commonly applied across a vast array of disciplines, such as economics, epidemiology, medicine, psychology,
sociology, and many others. Mediation analysis originated in the context of structural equation modeling, in which mediated effects were defined in terms of structural equation model coefficients \citep{wright1921correlation}. However, defining mediated effects in such a way has a number of limitations. Since such definitions are not connected to interventions or counterfactuals, they lack a clear causal interpretation. Additionally, this literature had previously made no explicit assumptions about, nor formal attempts to account for confounding. Lastly, mediated effects in this literature were defined with respect to a particular model, often a system of linear equations, hence they did not allow for interactions and nonlinearities. 

\cite{robins1992identifiability} and \cite{pearl2001direct} defined mediated effects in terms of nested counterfactuals, thereby grounding mediation in the formal language of causal inference. Following these landmark articles, causal inference has witnessed an explosion of literature on mediation analysis (see \cite{vanderweele2015explanation} for an overview). 
Despite this significant leap forward in the formalization of mediated effects, both their existence and identification have remained controversial. This article will focus primarily on the latter point rather than the first. \cite{pearl2001direct} observed that natural indirect effects (NIEs) are not identified in the presence of a confounder of the mediator--outcome relationship that is affected by the exposure. I will refer to such a variable as an \emph{exposure-induced confounder}; it is also referred to in the context of causal graphical models as a \emph{recanting witness}. This is troublesome, as the existence of an exposure-induced confounder cannot be eliminated even in a well-controlled randomized experiment. \cite{robins2003semantics} and \cite{robins2010alternative} highlighted the fact that one of the identification assumptions for the natural indirect effect involves the conditional independence of two counterfactuals that are defined with respect to two conflicting interventions. The latter referred to such an assumption as a \emph{cross-world counterfactual independence} assumption, since it involves two counterfactuals that live in different interventional ``worlds'', as it were. One problem with such an assumption is that 
it cannot be enforced experimentally without randomizing the mediator, which thereby renders the effect of exposure on the mediator unobservable. Another is that the assumption cannot be falsified.

More recently, attempts have been made to define mediated effects in such a way that circumvents the need for 
cross-world counterfactual independencies and the absence of exposure-induced confounders. These all involve defining a randomized intervention on the mediator that sets the mediator to a random draw from the distribution of the counterfactual mediator under an intervention on the exposure, rather than to the actual realized value of the counterfactual mediator for that subject itself. \cite{robins2003semantics} first discussed a causal parameter that can be interpreted as the mean outcome under such a randomized intervention, and highlighted that it is nonparametrically identifiable in the presence of an exposure-induced confounder and without any cross-world counterfactual independence assumptions. \cite{van2008direct} discussed direct and indirect effect contrasts of such causal parameters, and provided an alternative causal interpretation in terms of a weighted average of a controlled direct effect. In the intervening years, \cite{didelez2006direct} and \cite{geneletti2007identifying} 
defined direct and indirect effects with respect to randomized interventions on the mediator in the absence of an exposure-induced confounder, but without defining counterfactuals. While these effects are distinct from the natural direct and indirect effects of \cite{robins1992identifiability} and \cite{pearl2001direct} in their definition and interpretation, the authors equated the two under a causal model with no exposure-induced confounder. \cite{didelez2006direct}, \cite{geneletti2007identifying}, and \cite{van2008direct} all additionally defined generalizations of these effects by allowing for any user-specified intervention distribution on the mediator. \cite{vanderweele2014effect} defined what they termed ``randomized interventional analogs of the natural direct and indirect effects'', which are nearly identical to the direct and indirect effects of \cite{didelez2006direct} and \cite{geneletti2007identifying}, with two important distinctions. First, \cite{vanderweele2014effect} defined their effects in terms of counterfactuals, and second, they defined them in the presence of an exposure-induced confounder, noting that these effects remain nonparametrically identified in such settings. \cite{vanderweele2017mediation} and \cite{zheng2017longitudinal} introduced two distinct extensions of the effect of \cite{vanderweele2014effect} to a longitudinal, time-varying exposure setting. In such settings, the identification challenge posed by exposure-induced confounders is exacerbated, as discussed in \cite{shpitser2013counterfactual}, and while certain path-specific effects can be identified, 
analogs of the NIE with time-varying exposures are extremely unlikely to be.

Since \cite{vanderweele2014effect}, randomized interventional definitions of direct and indirect effects have gained a great deal of popularity in the causal inference literature \citep{lok2016defining, lin2017interventional, 
vansteelandt2017interventional, 
vansteelandt2019mediation,mittinty2020longitudinal, diaz2021nonparametric, 
nguyen2021clarifying, 
xia2021identification, devick2022role, loh2022nonlinear, rudolph2022efficiently}. I shall focus on the randomized interventional indirect effect of \cite{vanderweele2014effect}, which I will denote NIE$^{\text{R}}$, as this formulation appears to have had the most widespread adoption; however, I reiterate that it is nearly identical to the effects of \cite{didelez2006direct}, \cite{geneletti2007identifying}, and \cite{van2008direct}.

There has been some folk wisdom circulating (dispensed primarily through comments during conference sessions and personal communications) that these randomized interventional indirect effects do not capture a true mediational effect. However, to the best of my knowledge, this notion has not been clearly formalized nor officially documented in the mediation analysis literature. In this article, I propose 
criteria for a definition of an indirect effect to have a true mediational interpretation, and demonstrate that the NIE$^{\text{R}}$ does not satisfy these criteria without stronger assumptions than have been given for its identification. I give examples of such assumptions, and show that for most of these, the standard NIE will also be identified, hence the NIE$^{\text{R}}$ no longer offers a clear advantage over the NIE as an indirect effect measure in such settings.

The remainder of the article is organized as follows. In Section \ref{sec:prelims}, I define notation, causal estimands, causal models, and provide nonparametric identification formulas for the causal estimands. In Section \ref{sec:snc}, I define desiderata for indirect effect measures. 
In Sections \ref{sec:recanting-witness} and \ref{sec:no-recanting}, I discuss satisfaction of the indirect effect measure criteria defined in Section \ref{sec:snc} in the presence of an exposure-induced confounder and in the absence of cross-world counterfactual independencies, respectively. In Section \ref{sec:interpretations}, I discuss alternative, non-mediational causal interpretations of randomized interventional indirect effects. In Section \ref{sec:related}, I discuss other randomized interventional indirect effects that are closely related to the NIE$^{\text{R}}$. In Section \ref{sec:pepfar}, I illustrate the findings in this article with an application to the question of whether adherence mediates the effect of antiretroviral therapy on virologic failure among HIV patients in data from the Harvard PEPFAR program in Nigeria. Lastly, I conclude with a discussion in Section \ref{sec:discussion}.


\section{Preliminaries}
\label{sec:prelims}

\subsection{Notation and estimands}
\label{subsec:estimands}
In the mediation analysis setting, one observes (at a minimum) i.i.d.~longitudinal samples of the random variables $(A,M,Y)$, where $A$ is the exposure, $Y$ is the outcome measured at a follow-up timepoint, and $M$ is a potential mediator measured at an intermediate timepoint between $A$ and $Y$. I will consider $A$ to be binary taking values in $\{a^*,a\}$ in this article unless otherwise stated; however, results will extend naturally to general real-valued exposures. I will use $a'$ to denote an arbitrary value in $\{a^*,a\}$. The level $a^*$ is often used to denote ``untreated'' or a control condition such as standard of care, and $a$ to denote ``treated''. More generally these can be thought of as a (possibly arbitrarily chosen) ``reference level'' and a ``comparison level'', respectively, as is the case in the PEPFAR example. 

For a motivating example, I will illustrate with a mediation analysis considered in \cite{miles2017quantifying,miles2020semiparametric}, in which we studied whether adherence to a prescribed antiretroviral therapy (ART) regimen plays a role in mediating the effect of different ART regimens on virologic failure among HIV patients observed in the Harvard President's Emergency Plan for AIDS Relief (PEPFAR) program in Nigeria. This was an observational study with patients taking a number of different ART regimens. We focused on two regimens in particular, defining the exposure $A$ to be an indicator of ART regimen, with $a^*=\text{TDF+3TC/FTC+NVP}$ and $a=\text{AZT+3TC+NVP}$. 
The mediator $M$ was a trichotomization of a continuous measure of adherence based on pill counts averaged over six months, with cutoffs at 80\% and 95\%. The outcome $Y$ was an indicator of virologic failure at one year, defined as repeat viral load above 1,000 copies/mL at 12 and 18 months. 

I will first assume the existence of counterfactual outcomes (or potential outcomes) $Y(a')$, which is the outcome we would have observed (possibly contrary to fact) had $A$ been assigned to level $a'$. In the PEPFAR example, this would be the virologic failure status we would have observed had we intervened to assign the patient to receive ART regimen $a'$.

Throughout, I will focus on effect measures on the difference scale, though results generalize naturally to other scales as well. The \emph{average treatment effect} (ATE), also often referred to as the \emph{total effect} (TE) in the context of mediation analysis, is defined as $\text{TE}\equiv E\{Y(a)\}-E\{Y(a^*)\}$, and is the difference in expected counterfactual outcomes under an intervention assigning all units to the comparison level versus an intervention assigning all units to the reference level. In the PEPFAR example, the total effect is the difference in risk of virologic failure had patients been assigned to AZT+3TC+NVP versus TDF+3TC/FTC+NVP. One desirable property of an indirect effect (or collection of path-specific effects) is that they decompose the total effect into a direct effect and an indirect effect (or collection of path-specific effects).

To discuss mediated effects, we must also assume the existence of the counterfactual $Y(a',m)$ for each $a'\in\{a^*,a\}$, which is the outcome we would have observed (possibly contrary to fact) had $A$ been assigned to level $a'$ and $M$ been assigned to level $m$. \cite{robins1992identifiability} and \cite{pearl2001direct} defined the \emph{controlled direct effect} (CDE) as $\text{CDE}(m)\equiv E\{Y(a,m)\}-E\{Y(a^*,m)\}$, which is the difference in expected counterfactual outcomes under an intervention assigning all units to be exposed versus an intervention assigning all units to the control condition while also intervening to set the mediator to level $m$ in both cases. In the PEPFAR example, the $\text{CDE}(m)$ is the difference in risk of virologic failure had patients been assigned to AZT+3TC+NVP versus TDF+3TC/FTC+NVP when also forcing patients to adhere at level $m$. In general, there is not considered to be an indirect effect corresponding to the controlled direct effect for reasons to be discussed in the following section. However, the \emph{portion eliminated}, defined as $\text{PE}(m)\equiv\text{TE} - \text{CDE}(m)$, has a useful policy interpretation as the portion of the effect of $A$ on $Y$ that would be eliminated were $M$ to be intervened on to be set to the level $m$, despite lacking a mediation interpretation \citep{robins1992identifiability}. In the PEPFAR example, the $\text{PE}(m)$ is the portion of the total effect of ART regimen on risk of virologic failure that would be eliminated were patients forced to adhere at level $m$. The portion eliminated divided by the total effect is known as the \emph{proportion eliminated}. 

To discuss natural direct and indirect effects, we must further assume the existence of the counterfactual mediator (or potential mediator) $M(a')$ for each $a'\in\{a^*,a\}$ and the nested counterfactual $Y\{a',M(a'')\}$ for each $(a',a'')\in\{a^*,a\}^2$. The former is the mediator we would have observed (possibly contrary to fact) had $A$ been assigned to level $a'$. The latter is interpreted as the outcome we would have observed (possibly contrary to fact) had $A$ been assigned to $a'$ and $M$ been assigned to its counterfactual value $M(a'')$ under an intervention setting $A$ to $a''$. The composition assumption states that $Y\{a',M(a')\}=Y(a')$ for each $a'\in\{a^*,a\}$. \cite{robins1992identifiability} define the contrast $E[Y\{a,M(a)\}]-E[Y\{a,M(a^*)\}]$ to be the total indirect effect, and the contrast $E[Y\{a,M(a^*)\}]-E[Y\{a^*,M(a^*)\}]$ to be the pure direct effect. \cite{pearl2001direct} defined the latter to be the natural direct effect 
and $E[Y\{a^*,M(a)\}]-E[Y\{a^*,M(a^*)\}]$ 
to be the natural indirect effect (or the negative of the total indirect effect when instead $a$ is considered to be the exposure reference level). In this article, I will refer to $E[Y\{a,M(a)\}]-E[Y\{a,M(a^*)\}]$ as the \emph{natural indirect effect} (NIE) and to $E[Y\{a,M(a^*)\}]-E[Y\{a^*,M(a^*)\}]$ as the \emph{natural direct effect} (NDE), as is common in the mediation literature. In the PEPFAR example, the NIE is the change in risk of virologic failure when patients are assigned to receive AZT+3TC+NVP, but their adherence level switches from that under AZT+3TC+NVP to what it naturally would have been had they instead been assigned to receive TDF+3TC/FTC+NVP. The NDE is the difference in risk of virologic failure had patients been assigned to AZT+3TC+NVP versus TDF+3TC/FTC+NVP, but their adherence remained at the level it naturally would have taken had they been assigned to TDF+3TC/FTC+NVP in either case. Given the composition assumption, we have the following effect decomposition: 
\begin{align*}
    \text{TE} &= E\{Y(a)\}-E\{Y(a^*)\} = E[Y\{a,M(a)\}]-E[Y\{a^*,M(a^*)\}]\\
    &= E[Y\{a,M(a)\}]-E[Y\{a,M(a^*)\}] + E[Y\{a,M(a^*)\}]-E[Y\{a^*,M(a^*)\}]\\ 
    &= \text{NIE} + \text{NDE}.
\end{align*}

Lastly, I will review definitions of the randomized interventional analogs to the natural direct and indirect effects, also commonly referred to as randomized interventional direct and indirect effects or simply interventional direct and indirect effects. These are defined with respect to a distinct hypothetical random variable, denoted $G(a')$ (\cite{vanderweele2014effect} use the notation $G(a'\mid {\bf C})$), such that it is closely related to, yet distinct from $M(a')$. 
The random variable $G(a')$ has the following two properties. First, it follows the same conditional distribution given some set of pre-exposure covariates ${\bf C}$, as the counterfactual mediator $M(a')$, i.e., $F_{G(a')\mid {\bf C}}(m\mid {\bf c}) = F_{M(a')\mid {\bf C}}(m\mid {\bf c})$, where each is the corresponding conditional distribution function given ${\bf C}$. Second, $G(a')$ is randomized within strata of ${\bf C}$ according to the distribution of $M(a')$ such that it is conditionally independent of all other observed and counterfactual (nested or otherwise) variables that have been defined thus far given ${\bf C}$. Technically, it is only the conditional independence of $G(a')$ and $Y(a,m)$ given ${\bf C}$ for each $a'\in\{a,a^*\}$ that is needed for nonparametric identification, though conditional randomization of $G(a')$ makes for a simpler interpretation. In fact, the joint conditional distribution of $G(a')$ with $Y(a,m)$ given $\bf C$ may differ from the joint conditional distribution of $M(a')$ and $Y(a,m)$ given $\bf C$, the latter of which may well not be independent, even though $G(a')$ and $M(a')$ share the same conditional distributions given $\bf C$. As an example, suppose the true conditional joint distribution of $M(a')$ and $Y(a,m)$ given $\bf C$ is bivariate normal with mean $({\bf C^T}\boldsymbol{\beta}_M, {\bf C^T}\boldsymbol{\beta}_Y)^T$ for some parameters $\boldsymbol{\beta}_M$ and $\boldsymbol{\beta}_Y$ and covariance $\boldsymbol\Sigma$ with diagonal elements equal to one and nonzero off-diagonal. Then for subject $i$, their $G_i(a')$ would be a random draw from the distribution $N({\bf C_i}^T\boldsymbol{\beta},1)$; however, the conditional joint distribution of $G(a')$ and $Y(a,m)$ given $\bf C$ would be bivariate normal with mean $({\bf C^T}\boldsymbol{\beta}_M, {\bf C^T}\boldsymbol{\beta}_Y)^T$, but with identity covariance matrix.

Let $Y\{a',G(a'')\}$ be the counterfactual outcome we would have observed (possibly contrary to fact) had $A$ been assigned to $a'$ and $M$ been assigned to be $G(a'')$, i.e., a value randomly drawn from the conditional distribution of the counterfactual mediator $M(a'')$ given ${\bf C}$ (as opposed to the precise value realized by $M(a'')$ itself). In the PEPFAR example, $Y\{a',G(a'')\}$ is the virologic failure status we would have observed had the patient been assigned to ART regimen $a'$ and their adherence been forced to take a level determined by a random draw from the conditional distribution of what the adherence level would have been had they been assigned to ART regimen $a''$ given their baseline covariates. \cite{vanderweele2014effect} define the \emph{randomized interventional analog to the natural direct and indirect effects} to be $\text{NDE}^{\text{R}}\equiv E[Y\{a,G(a^*)\}] - E[Y\{a^*,G(a^*)\}]$ and $\text{NIE}^{\text{R}}\equiv E[Y\{a,G(a)\}] - E[Y\{a,G(a^*)\}]$, respectively. In the PEPFAR example, these effects can be interpreted as contrasts of the mean counterfactuals $Y\{a',G(a'')\}$ interpreted above for different specifications of $a'$ and $a''$. The interpretation of the $\text{NIE}^{\text{R}}$ in the PEPFAR example will be discussed in greater detail in Section \ref{sec:pepfar}. Unlike the NDE and NIE, the $\text{NDE}^{\text{R}}$ and $\text{NIE}^{\text{R}}$ do not decompose the total effect without assumptions going beyond those that are used for nonparametric identification, since $E[Y\{a',G(a')\}] = E[Y\{a',M(a')\}]$ does not hold in general. Instead, they decompose a version of a total effect that is defined with respect to the randomized intervention, viz., $\text{TE}^{\text{R}}\equiv E[Y\{a,G(a)\}] - E[Y\{a^*,G(a^*)\}]$. Clearly, this has a less straightforward interpretation than the TE, which contrasts much simpler interventions.


\subsection{Causal modeling assumptions and nonparametric identification}
\label{subsec:assns}

The various effects defined in Section \ref{subsec:estimands} can be nonparametrically identified under certain consistency, exchangeability, and positivity assumptions. Exchangeability assumptions rule out confounding as an alternative explanation for associations and ensure that subjects with one exposure and/or mediator realization can be used to estimate facets of the distribution of counterfactuals under this realization for a subject with a distinct exposure and/or mediator realization. Nonparametric identification results for the effects defined in Section \ref{subsec:estimands} each rely on some subset of exchangeability assumptions, each of which is implied by one or both of two causal models defined with respect to a causal directed acyclic graph (DAG): the nonparametric structural equation model with independent errors (NPSEM-IE) \citep{pearl1995causal} and the finest fully randomized causally interpretable structured tree graph (FFRCISTG) \citep{robins1986new}. The general definitions of these models are given in Section S1 of the supporting web materials. Here, I will first define the causal models with respect to a particular DAG, then state the relevant exchangeability assumptions that are implied by each under the corresponding DAG.

Consider first the DAG in Figure \ref{fig:DAG1}. 
\begin{figure}[h]
    \centering
    \begin{tikzpicture}[baseline={(A)}, 
    ->, line width=1pt]
    \tikzstyle{every state}=[draw=none]
    \node[shape=circle, draw] (C) at (0,0) {${\bf C}$};
    \node[shape=circle, draw] (A) at (2,0) {$A$};
    \node[shape=circle, draw] (M) at (4,0) {$M$};
    \node[shape=circle, draw] (Y) at (6,0) {$Y$};

      \path 	(C) edge (A)
                (C) edge  [bend right=30] (M)
                (C) edge  [bend right=45] (Y)
                (M) edge (Y)
                (A) edge (M)
                (A) edge  [bend left=45] (Y)
                          ;
    \end{tikzpicture}
    \caption{The standard mediation graph with a single intermediate variable, $M$.\label{fig:DAG1}}
\end{figure}
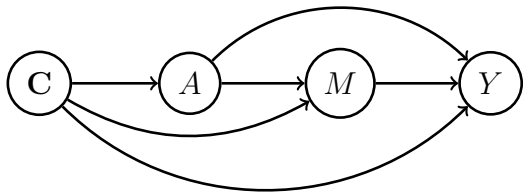
The NPSEM-IE corresponding this DAG is defined to be the system of structural equations: ${\bf C} = {\bf g_C}(\boldsymbol{\varepsilon}_{\bf C})$, $A = g_A({\bf C}, \varepsilon_A)$, $M = g_M({\bf C}, A, \varepsilon_M)$, and $Y = g_Y({\bf C}, A, M, \varepsilon_Y)$, where for each $V\in\{{\bf C}, A, M, Y\}$, $g_{V}$ is an unspecified (i.e., ``nonparametric'') function mapping to the domain of $V$, and $\varepsilon_{V}$ is an exogeneous error term. Further, the exogeneous error terms $\{\boldsymbol{\varepsilon}_{\bf C}, \varepsilon_A, \varepsilon_M, \varepsilon_Y\}$ are assumed to be mutually independent under the NPSEM-IE (hence ``independent errors''). Under a given intervention, these structural equations can be modified to generate equations for the counterfactuals under the intervention. First, the left-hand side of the equations for the descendants of the intervened-upon variables are replaced with their corresponding counterfactuals under the intervention. Second, the intervened-upon variables are replaced in the right-hand side of the structural equations with the level to which they have been set. Third, the descendants of the intervened-upon variables are replaced in the right-hand side of the structural equations with their corresponding counterfactual versions. For example, under the intervention setting $A$ to $a$, the structural equations would become ${\bf C} = {\bf g_C}(\boldsymbol{\varepsilon}_{\bf C})$, $A = a$, $M(a) = g_M({\bf C}, a, \varepsilon_M)$, and $Y(a) = g_Y\{{\bf C}, a, M(a), \varepsilon_Y\}$. In this way, it then becomes straightforward to determine which observed and counterfactual variables are independent under the NPSEM-IE. The NPSEM-IE can be equivalently defined as the set of all counterfactual probability distributions  for which the following sets of counterfactuals are all mutually independent: $\{{\bf C}\}$, $\{A({\bf c}) \mid {\bf c}\in\text{supp}({\bf C})\}$, $\{M({\bf c},a') \mid {\bf c}\in\text{supp}({\bf C}),a'\in\{a^*,a\}\}$, and $\{Y({\bf c},a',m) \mid {\bf c}\in\text{supp}({\bf C}),a'\in\{a^*,a\},m\in\text{supp}(M)\}$, where each of these counterfactuals is assumed to exist in addition to those described in Section \ref{subsec:estimands}.

The FFRCISTG corresponding to the DAG in Figure \ref{fig:DAG1} consists of all counterfactual probability distributions for which the counterfactuals in the set $\{{\bf C}, \allowbreak A({\bf c}), \allowbreak M({\bf c}, a'), \allowbreak Y({\bf c}, a', m)\}$ are all mutually independent for each $({\bf c}, a', m)\in \text{supp}({\bf C})\times\{a^*,a\}\times\text{supp}(M)$, but unlike the NPSEM-IE, are not necessarily independent across sets for different values of $({\bf c}, a', m)$. Independencies across such sets, e.g., $Y({\bf c},a,m)\ci M({\bf c},a^*)$, are known as \emph{cross-world counterfactual independencies}, which will be discussed further below. The independence assumptions imposed by the NPSEM-IE imply those of the FFRCISTG, hence the latter model contains the former.

Now, I will present the set of exchangeability assumptions implied by the above causal models that are used for the identification of the effects in the previous subsection. 
\begin{assumption}
    \label{assn:NUCA-AY}
    $Y(a',m)\ci A\mid {\bf C}$ for all $a'$ and $m$.
\end{assumption}
\begin{assumption}
    \label{assn:NUCA-MY}
    $Y(a',m)\ci M\mid {\bf C},A=a'$ for all $a'$ and $m$.
\end{assumption}
\begin{assumption}
    \label{assn:NUCA-AM}
    $M(a')\ci A\mid {\bf C}$ for all $a'$.
\end{assumption}
\begin{assumption}
    \label{assn:NPSEM}
     $Y(a,m)\ci M(a^*)\mid {\bf C}$ for all $m$.
\end{assumption}
The first three conditional independencies are implied by both the NPSEM-IE and the FFRCISTG corresponding to the DAG in Figure \ref{fig:DAG1}. These are relatively uncontroversial and correspond to typical exchangeability/no-unobserved-confounding assumptions provided all elements of ${\bf C}$ are unaffected by the exposure. Assumption \ref{assn:NUCA-AY} states that there is no unobserved confounding of the effect of $A$ on $Y$ given ${\bf C}$ when $M$ is set to the level $m$. Assumption \ref{assn:NUCA-MY} states that there is no unobserved confounding of the effect of $M$ on $Y$ given ${\bf C}$ and $A$. Assumption \ref{assn:NUCA-AM} states that there is no unobserved confounding of the effect of $A$ on $M$ given ${\bf C}$. Assumptions \ref{assn:NUCA-AY} and \ref{assn:NUCA-AM} will be enforced when $A$ is randomized. Assumption \ref{assn:NUCA-MY} will be enforced when both $A$ and $M$ are randomized; however, it is not sensible to jointly randomize $A$ and $M$ if an indirect effect is the effect of interest, as one would no longer be able to observe the relationship between $A$ and $M$, which is a key component of the indirect effect. 

The fourth conditional independence is implied by the NPSEM-IE corresponding to the DAG in Figure \ref{fig:DAG1}, but not the corresponding FFRCISTG. It is one of the main sources of controversy for indirect effects, for two distinct reasons. The first is that this assumption prohibits any of the elements of ${\bf C}$ from being affected by $A$. This means that all confounders of the effect of $M$ on $Y$ must either be pre-exposure or unaffected by $A$. While this assumption can be tested, it cannot be enforced experimentally, and so without strong alternative assumptions, the investigator has no control over whether this holds. The second is that this assumption is a so-called \emph{cross-world counterfactual independence} assumption, meaning that when $a'\neq a''$, $Y(a',m)$ and $M(a'')$ cannot be observed simultaneously for the same individual since these arise from conflicting interventions. As a consequence, this independence can neither be enforced experimentally, nor falsified empirically. \cite{robins2010alternative}, \cite{tchetgen2014bounds}, and \cite{miles2017partial} provide bounds on the natural indirect effect when this assumption is not satisfied. 

The consistency assumption links counterfactuals to observed variables and holds under both an NPSEM-IE and an FFRCISTG.
\begin{assumption}[Consistency]
    \label{assn:consistency}
    If $A=a'$, then $M=M(a')$ almost surely. If $A=a'$ and $M=m$, then $Y=Y(a',m)$ almost surely.
\end{assumption}

The positivity assumption ensures that the identifying functionals are well-defined and that there is statistical support to estimate such quantities. For simplicity, we will consider the strongest form of the positivity assumption; however, this can be relaxed for some identification results. 
\begin{assumption}[Positivity]
    \label{assn:positivity}
    We have $0< \mathrm{Pr}(A=a\mid {\bf C})<1$ almost surely, 
    and for all $a'\in\{a^*,a\}$ and all $m$ in either the support of $M$ or the support of $M(a')$, $f_{M\mid {\bf C},A}(m\mid a',{\bf C})>0$ almost surely.
\end{assumption}
Unlike Assumptions \ref{assn:NUCA-AY}--\ref{assn:consistency}, the positivity assumption is only stated in terms of the observed data distribution. Of course, this implies constraints on the counterfactual distribution due to the connection between the two under the consistency assumption. While the positivity assumption for the total effect (also referred to as overlap) has received a fair amount of attention in the literature, unfortunately the same cannot be said of this assumption for mediated effects. As positivity for mediation is not the focus of this article, we direct the reader to \cite{nguyen2022clarifying} for a more in-depth discussion.

Under a slightly weaker form of Assumptions \ref{assn:NUCA-AY}, \ref{assn:consistency}, and \ref{assn:positivity}, the total effect is nonparametrically identified by
\[\Psi^{TE}(P) \equiv E\left\{E(Y\mid a, {\bf C})\right\} - E\left\{E(Y\mid a^*, {\bf C})\right\},\]
where $P$ denotes the distribution function of the observed data $({\bf C}^T,A,M,Y)^T$. Under Assumptions \ref{assn:NUCA-AY}, \ref{assn:NUCA-MY}, \ref{assn:consistency}, and \ref{assn:positivity}, the controlled direct effect is nonparametrically identified by 
\[\Psi^{CDE}(P;m) \equiv E\left\{E(Y\mid m, a, {\bf C})\right\} - E\left\{E(Y\mid m, a^*, {\bf C})\right\},\]
and the portion eliminated by $\Psi^{PE}(P;m) \equiv \Psi^{TE}(P) - \Psi^{CDE}(P;m)$. 
Under Assumptions \ref{assn:NUCA-AY}--\ref{assn:positivity}, the natural indirect effect is nonparametrically identified by
\[\Psi^{NIE}(P) \equiv E\left\{E\left(Y\mid a, {\bf C}\right)\right\} - E\left[E\left\{E\left(Y\mid M, a, {\bf C}\right)\mid a^*, {\bf C}\right\}\right],\]
which is often referred to as the mediation formula, and the natural direct effect is nonparametrically identified by $\Psi^{NDE}(P) \equiv \Psi^{TE}(P) - \Psi^{NIE}(P)$. Under Assumptions \ref{assn:NUCA-AY}--
\ref{assn:positivity}, the NIE$^\text{R}$ is also nonparametrically identified by $\Psi^{NIE}(P)$. If all of these assumptions except Assumptions \ref{assn:NUCA-MY} and \ref{assn:NPSEM} hold, and Assumption \ref{assn:NUCA-MY} is replaced by an analogous assumption allowing for exposure-induced confounders (Assumption \ref{assn:NUCA-MY-L} in Section \ref{sec:recanting-witness}), the NIE$^\text{R}$ remains nonparametrically identified, but by a distinct formula, which is given in Section \ref{sec:recanting-witness}.

\section{Formalizing desiderata of indirect effect measures}
\label{sec:snc}

When considering total effects, the sharp causal null on a given population of interest (e.g., the sample population or superpopulation from which the sample was drawn) is defined to be the null hypothesis $H_0:Y_i(a)=Y_i(a^*)$ for all units $i$ in the population, i.e., there is no individual-level causal effect for any unit in the population. In defining total effect measures, it is self-evident that such measures should take their null value (e.g., zero on the difference scale) whenever the sharp causal null holds on the population of interest. Indeed, this is the case for the total effect. I will now extend this notion to measures of indirect effects.

First, we must define the sharp mediational null. To do so, we must formalize what we mean when we speak of mediation for an individual unit. 
I would argue that the statement ``the effect of an exposure $A$ on an individual's outcome $Y$ is mediated by an intermediate event $M$'' is widely understood 
to mean that 
(a) an intervention changing $A$ leads to a change in $M$, and (b) this induced change in $M$ leads to a change in $Y$. This interpretation of mediation is centered on the question of whether $M$ plays a role in the \emph{mechanism} by which $A$ affects $Y$. Indeed, the introductory chapter of \cite{vanderweele2015explanation} states, ``The phenomenon whereby a cause affects an intermediate and the change in the intermediate goes on to affect the outcome is what is generally referred to as the phenomenon of `mediation'[...]'' Thus, practitioners using effects whose interpretations differ from this characterization of mediation should either avoid interpreting such effects as mediational, or should be explicit that their characterization of mediation differs from this predominant understanding, and clarify what they mean when using the term ``mediation''. 

Formalizing this in terms of counterfactuals for a given unit $i$, (a) amounts to $M_i(a)\neq M_i(a^*)$, and the weakest form of (b) would be $Y_i\{a',M_i(a)\}\neq Y_i\{a',M_i(a^*)\}$ for either $a'=a$ or $a'=a^*$. Clearly, the latter set of inequations implies the former, so these suffice on their own for existence of an individual-level indirect effect. This corresponds to the qualitative definition of unit-level indirect effects of \cite{pearl2001direct} for at least one choice of the exposure reference level. That is, if there is a qualitative unit-level indirect effect according to Pearl for one reference exposure level, say $a^*$ (but not necessarily the other), we would claim there is an indirect effect at the individual level. Conversely, there is no individual-level indirect effect when $Y_i\{a',M_i(a)\} = Y_i\{a',M_i(a^*)\}$ for both $a'=a$ and $a'=a^*$. I define the sharp mediational null hypothesis for a given population as follows:
\begin{definition}[Sharp mediational null]
    $H_0: Y_i\{a',M_i(a)\} = Y_i\{a',M_i(a^*)\}$ for both $a'=a$ and $a'=a^*$ for each $i$ in the population of interest.
\end{definition}
While it is necessary for $A$ to affect $M$ and $M$ to affect $Y$ for there to be mediation at the individual level, this alone is not sufficient according to this definition. For instance, we could have trichotomous $M$ with support $\{1,2,3\}$, and for subject $i$, $M_i(a)=1$, $M_i(a^*)=2$, and $Y_i(a,1)=Y_i(a^*,1)=Y_i(a,2)=Y_i(a^*,2)$, but $Y_i(a,1)\neq Y_i(a,3)$. Thus, in this example $A$ affects $M$ and $M$ affects $Y$ (between $M=1$ and $3$ and between $M=2$ and $3$) for unit $i$, but the change in $M$ induced by the change in $A$ does not induce a change in $Y$. In this example, $Y_i(a,3)$ would not be naturally observed, and would only be realized for subject $i$ under a joint intervention setting $A_i$ to $a$ and $M_i$ to 3. This is analogous to how if $A_i=a$, then $Y_i(a^*)$ would only be realized under an intervention setting $A_i$ to $a^*$. 

If one would consider this example to in fact be a case of mediation at the individual level, then one would arrive at a distinct sharp null. That is, one might instead interpret mediation as meaning that (a) an intervention changing $A$ leads to a change in $M$, and (b) an intervention changing $M$ leads to a change in $Y$, rather than the change in $M$ in (b) needing to be induced by the change in $A$ in (a). I define the corresponding ``sharper'' mediational null as follows:
\begin{definition}[Sharper mediational null]
    $H_0$: For each $i$ in the population of interest, either $M_i(a)=M_i(a^*)$ or $Y_i(a',m) = Y_i(a',m')$ for all $a'$, $m$, and $m'$.
\end{definition}
If one were to object to the existence of so-called cross-world counterfactuals, i.e., nested counterfactuals of the form $Y_i\{a',M_i(a'')\}$ with $a'\neq a''$, then one might prefer the above definition. Since $Y_i(a',m) = Y_i(a',m')$ must hold for all $m$ and $m'$ when $M_i(a)\neq M_i(a^*)$ and not just for $m=M_i(a)$ and $m'=M_i(a^*)$, clearly the sharper mediational null implies the sharp mediational null, hence the term ``sharper''.


I now define the mediational analogs to the previously-described desideratum of total effect measures taking their null value under the sharp(er) causal null.
\begin{definition}[Sharp(er) null criterion]
    An indirect effect measure satisfies the \emph{sharp(er) null criterion} if it is null whenever the sharp(er) mediational null holds.
\end{definition}
Since the class of distributions satisfying the sharp mediational null contains the class of distributions satisfying the sharper mediational null, an effect measure that satisfies the sharp null criterion will also satisfy the sharper null criterion.

Unfortunately, individual-level indirect effects, like individual-level total/overall effects, are never identified. This is due to the fact that both counterfactual mediators cannot be observed simultaneously for the same subject, hence the value of $M(a')$ will be unknown for at least one of the two nested counterfactuals. Thus, population-level effects such as those defined in Section \ref{sec:prelims} are generally the causal contrasts of interest in mediation analysis. 

The NIE satisfies the sharp null criterion (and, hence, the sharper null criterion) since under the mediational sharp null, $Y_i\{a, M_i(a)\} = Y_i\{a, M_i(a^*)\}$ for all $i$, hence $E[Y\{a, M(a)\}] - E[Y\{a, M(a^*)\}]=0$. Since the NIE$^\text{R}$ is equal to the NIE under Assumptions \ref{assn:NUCA-AY}--\ref{assn:positivity}, the NIE$^\text{R}$ will also satisfy the sharp null criterion (and, hence, the sharper null criterion) when these assumptions all hold. On the other hand, $\text{PE}(m)$ (defined in Section \ref{sec:prelims}.1) satisfies neither the sharp null nor the sharper null criterion under its corresponding identifying Assumptions \ref{assn:NUCA-AY}, \ref{assn:NUCA-MY}, \ref{assn:consistency}, and \ref{assn:positivity}, 
which is why it is not considered to be a valid indirect effect measure. The reason it does not satisfy these criteria is because when there is interaction between $A$ and $M$ on the additive scale in their effect on $Y$, the controlled direct effect and total effect will differ for certain levels of $m$ even when $A$ has no effect on $M$ for any unit \citep{vanderweele2009mediation}. For instance, consider the following simple example: for each $a'\in\{a^*,a\}$, $M(a')\sim \mathrm{Bernoulli}(p)$ with $0<p<1$ such that $M_i(a^*)=M_i(a)$ for all $i$, and $Y(a',m)=a'm$. Then $\text{TE}=(a-a^*)p$ and $\text{CDE}(m)=(a-a^*)m$, so $\text{PE}(m)=(a-a^*)(p-m)$, which is nonzero for both $m=0$ and $m=1$ even though $A$ does not affect $M$ for any unit. Essentially, the portion eliminated 
can detect mediation \emph{or} interaction, but cannot distinguish between the two.

In addition to the sharp(er) null criterion, which ensures an indirect effect measure correctly detects when a mediational effect does not exist, one might also demand of an indirect effect measure that it be in the correct direction when individual-level mediational effects all agree in direction, i.e., qualitatively. Analogously to the monotonicity assumption often invoked to identify the local average treatment effect using instrumental variables, we may also define a mediational version of monotonicity. When an instrumental variable, say $Z$, is available, the effect of $Z$ on $A$ is said to be \emph{monotonic} if for all $z^*<z$, either $A_i(z^*)\leq A_i(z)$ for all $i$ or $A_i(z^*)\geq A_i(z)$ for all $i$, where $A_i(z')$ is the counterfactual exposure under an intervention setting $Z_i$ to $z'$ \citep{imbens1994identification}. One could likewise define the effect of $A$ on $Y$ to be \emph{monotonic} if $Y_i(a^*)\leq Y_i(a)$ for all $i$ or $Y_i(a^*)\geq Y_i(a)$ for all $i$. I now define a mediational analog to the monotonicity assumption.
\begin{definition}[Mediational monotonicity]
    A counterfactual distribution satisfies \emph{mediational monotonicity} if for each $a'\in\{a^*,a\}$, either (a) $Y_i\{a',M_i(a)\} \leq Y_i\{a',M_i(a^*)\}$ for each $i$ in the population of interest or (b) $Y_i\{a',\allowbreak M_i(a)\}\allowbreak \geq\allowbreak Y_i\{a',\allowbreak M_i(a^*)\}$ for each $i$ in the population of interest.
\end{definition}

The following property formalizes the notion that an indirect effect measure ought to be in the correct direction under mediational monotonicity.
\begin{definition}[Monotonicity criterion]
    \label{def:mono-criterion}
    An indirect effect measure satisfies the \emph{monotonicity criterion} if it is no greater than its null value (e.g., zero for an effect measure on the difference scale) whenever version (a) of mediational monotonicity holds and no less than its null value whenever version (b) of mediational monotonicity holds.
\end{definition}
Clearly, the NIE satisfies the monotonicity criterion since under monotonicity, either $Y_i\{a,M_i(a)\} \leq Y_i\{a,M_i(a^*)\}$ for all $i$, in which case $E[Y\{a,M(a)\}] \allowbreak - E[Y\{a,M(a^*)\}] \allowbreak \leq \allowbreak 0$, or $Y_i\{a,M_i(a)\} \geq Y_i\{a,M_i(a^*)\}$ for all $i$, in which case $E[Y\{a,M(a)\}] - E[Y\{a,M(a^*)\}] \geq 0$. An effect measure satisfying the monotonicity criterion will also satisfy the sharp null criterion (and in turn, the sharper null criterion) since the sharp mediational null implies both (a) and (b) in the definition of mediational monotonicity, and the monotonicity criterion implies the effect can be neither less than nor greater than its null value. 

For the remainder of the article, I will refer to the sharp null, sharper null, and monotonicity criteria collectively as the \emph{indirect effect measure criteria}, though these are not meant to be exhaustive, and additional desiderata could certainly be defined. In the following two sections, I will discuss whether the NIE$^\text{R}$ satisfies these criteria 
under alternative sets of assumptions.

\section{Indirect effects in the presence of an exposure-induced confounder}
\label{sec:recanting-witness}
One of the implications of Assumption \ref{assn:NPSEM} is the proscription of the existence of an exposure-induced confounder. However, this assumption is not necessary for the nonparametric identification of the NIE$^{\text{R}}$. We will now consider satisfaction of the indirect effect measure criteria under an NPSEM-IE in the presence of an exposure-induced confounder ${\bf L}$ as depicted in the DAG in Figure \ref{fig:DAG2}.
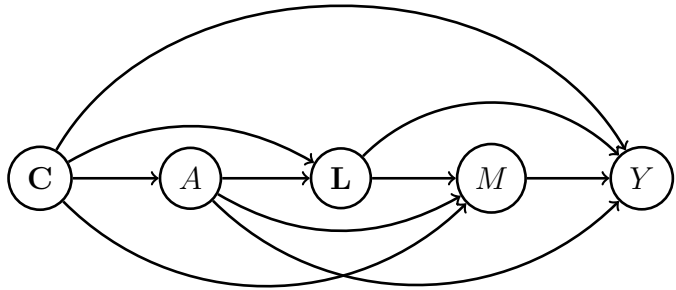
\begin{figure}[h]
    \centering
    \begin{tikzpicture}[baseline={(A)}, 
    ->, line width=1pt]
    \tikzstyle{every state}=[draw=none]
    \node[shape=circle, draw] (C) at (-2,0) {${\bf C}$};
    \node[shape=circle, draw] (A) at (0,0) {$A$};
    \node[shape=circle, draw] (L) at (2,0) {${\bf L}$};
    \node[shape=circle, draw] (M) at (4,0) {$M$};
    \node[shape=circle, draw] (Y) at (6,0) {$Y$};

      \path 	(C) edge (A)
                (C) edge [bend left=30] (L)
                (C) edge [bend right=45] (M)
                (C) edge [bend left=60] (Y)
                (A) edge (L)
                (A) edge [bend right=30] (M)
                (A) edge [bend right=45] (Y)
                (M) edge (Y)
                (L) edge (M)
                (L) edge  [bend left=45] (Y)
                          ;
    \end{tikzpicture}
    \caption{The mediation DAG with an exposure-induced confounder, ${\bf L}$.\label{fig:DAG2}}
\end{figure}

The NPSEM-IE corresponding to the DAG in Figure \ref{fig:DAG2} is defined to be the system of structural equations ${\bf C} = {\bf g_C}(\boldsymbol{\varepsilon}_{\bf C})$, $A = g_A({\bf C}, \varepsilon_A)$, ${\bf L} = {\bf g_L}({\bf C}, A, \boldsymbol{\varepsilon_L})$, $M = g_M({\bf C}, A, {\bf L}, \varepsilon_M)$, and $Y = g_Y({\bf C}, A, {\bf L}, M, \varepsilon_Y)$, where the functions $g_{V}$ and exogeneous error terms $\boldsymbol{\varepsilon}_V$ are analogously defined as in the NPSEM-IE corresponding the DAG in Figure \ref{fig:DAG1}. 
Assumptions \ref{assn:NUCA-AY} and \ref{assn:NUCA-AM} are still implied under this model, but Assumption \ref{assn:NUCA-MY} is not. Instead, it can be replaced by the following assumption, which does hold under this model:
\begin{assumption}
    \label{assn:NUCA-MY-L}
    $Y(a',m)\ci M\mid {\bf L}, {\bf C},A=a'$ for all $a'$ and $m$.
\end{assumption}
While Assumption \ref{assn:NPSEM} does not hold under this model, other cross-world counterfactual assumptions are implied. Under this model 
and Assumption \ref{assn:positivity}, the NIE$^{\mathrm{R}}$ is nonparametrically identified by
\begin{align*}
    \Psi^{\text{NIE}^{\text{R}}}_{L}(P)\equiv E\bigg[\int_m \int_{\boldsymbol{\ell}} & E(Y\mid m,\boldsymbol{\ell},a,{\bf C})dF_{{\bf L}\mid A, {\bf C}}(\boldsymbol{\ell}\mid a, {\bf C})\\
    &\times \{dF_{M\mid A,{\bf C}}(m\mid a,{\bf C}) - dF_{M\mid A,{\bf C}}(m\mid a^*,{\bf C})\}\bigg].
\end{align*}
When Assumption \ref{assn:NPSEM} does hold, this reduces to $\Psi^{\text{NIE}}(P)$. 

Despite being identified, we have the following theorem regarding the indirect effect measure criteria:
\begin{theorem}
    \label{thm:recanting-SNC}
    Under the NPSEM-IE corresponding to the DAG in Figure \ref{fig:DAG2} and Assumption \ref{assn:positivity}, the NIE$^{\mathrm{R}}$ does not satisfy any of the indirect effect measure criteria. 
\end{theorem}
All proofs are provided in Section S4 of the supporting web materials. The proof of Theorem \ref{thm:recanting-SNC} gives a counterexample of a counterfactual distribution satisfying the NPSEM-IE corresponding to the DAG in Figure \ref{fig:DAG2} as well as the sharp(er) mediational null, but under which the NIE$^{\mathrm{R}}$ is nonzero. While the exact details of the counterexample are unimportant, it does have some salient features that are helpful for understanding why the NIE$^{\mathrm{R}}$ does not satisfy the sharp(er) null criterion and how far from its null value it can be under the sharp(er) mediational null. First, the exogeneous error term corresponding to $L$, $\varepsilon_L$, is a Bernoulli random variable. When $\varepsilon_L=0$, $M$ affects $Y$, but there is no effect of $A$ on $M$, i.e., $M(a^*)=M(a)$. When $\varepsilon_L=1$, $A$ affects $M$; however, $M$ has no effect on $Y$, i.e., $Y(a,1)=Y(a,0)$ and $Y(a^*,1)=Y(a^*,0)$. 
Thus, while $A$ will affect $M$ for some units, and $M$ will affect $Y$ for others, there is no unit for whom both occur, 
and so the sharper mediational null holds, which in turn implies the sharp mediational null holds. This is a consequence of the presence of an interaction between $L$ and $A$ on $M$ and between $L$ and $M$ on $Y$. The other salient feature of the counterexample is that $Y$ is a Bernoulli random variable, and depending on the probability parameters of $\varepsilon_L$ and $\varepsilon_M$, which are both Bernoulli, the NIE$^{\mathrm{R}}$ can be arbitrarily close to -1/4 or 1/4, as shown in the proof. 

Hypothetically, if ${\bf L}$ were not affected by $A$, ${\bf L}$ would necessarily be included in ${\bf C}$ in order to satisfy Assumption \ref{assn:NUCA-MY}, and the NIE$^{\mathrm{R}}$ would in fact satisfy the indirect effect measure criteria. This suggests that when $A$ does in fact affect $\bf L$, one might be able to recover an identifiable effect satisfying the indirect effect measure criteria by replacing the intervention setting $M$ to $G(a')$ with an intervention setting $M$ to either (a) $G(a'\mid {\bf C}, {\bf L})$, where $G(a'\mid {\bf C}, {\bf L})\sim M(a')\mid {\bf C}, {\bf L}$ and $G(a'\mid {\bf C}, {\bf L})\ci Y(a,m)\mid {\bf C}, {\bf L}$, or (b) $G\{a'\mid {\bf C}, {\bf L}(a')\}$, where $G\{a'\mid {\bf C}, {\bf L}(a')\}\sim M(a')\mid {\bf C}, {\bf L}(a')$ and $G\{a'\mid {\bf C}, {\bf L}(a')\}\ci Y(a,m)\mid {\bf C}, {\bf L}(a')$. Section S2 of the supporting web materials presents results showing that this is unfortunately not the case for either intervention.

One might take the view that the counterexample in the proof of Theorem \ref{thm:recanting-SNC} is a contrived pathological example and unlikely to arise in practice. Indeed, \cite{vanderweele2017mediation} appear to have anticipated such a counterexample:
\begin{quote}
    [...]when natural direct and indirect effects are not identified, it will only be in extremely unusual settings that the interventional analogue is non-zero, with there being no natural indirect effects. For that to occur, it would be necessary that the exposure affects the mediator for a set of individuals that is completely different from those for whom the mediator affects the outcomes, i.e., there is no overlap in those for whom exposure affects the mediator and for whom the mediator affects the outcome.
\end{quote}
However, if one wishes to rule out such a counterexample, one must impose further assumptions on the model. Indeed, there exist assumptions under which the NIE$^{\mathrm{R}}$ does satisfy the sharp null criterion. A trivial example of such an assumption is that any of the edges connected to ${\bf L}$ apart from the edge from ${\bf C}$ to ${\bf L}$ in the DAG in Figure \ref{fig:DAG2} are not present, such that ${\bf L}$ is no longer an exposure-induced confounder, as was discussed in Section \ref{sec:prelims}.

Another such assumption is suggested by the remark from \cite{vanderweele2017mediation} above, viz., an assumption that rules out the possibility of no overlap in those for whom exposure affects the mediator and for whom the mediator affects the outcome. The following result states that while one formulation of such an assumption does recover satisfaction of the sharp(er) null criterion by the NIE$^{\text{R}}$, it still fails to satisfy the monotonicity criterion.
\begin{theorem}
    \label{thm:overlap}
    Suppose that if $M_i(a)\neq M_i(a^*)$ for some $i$ and $Y_j(a,m)\neq Y_j(a,m')$ for some $j$ and $m\neq m'$, 
    then $Y_{k}\{a,M_k(a^*)\}\neq Y_{k}\{a,M_k(a)\}$ for some $k$. 
    Then under the NPSEM-IE corresponding to the DAG in Figure \ref{fig:DAG2} and Assumption \ref{assn:positivity}, the $\text{NIE}^{\mathrm{R}}$ satisfies the sharp null and sharper null criteria, but not the monotonicity criterion.
\end{theorem}
Thus, if one were to rule out such a lack of overlap in the groups of individuals for whom each effect occurs, the $\text{NIE}^{\mathrm{R}}$ would indeed be useful for detecting the presence of an indirect effect. However, one could still not reliably learn the direction of such an effect, even if it were uniformly in the same direction (or null) for all individuals.

Section S3.1 of the supporting web materials gives an additional setting, wherein there is no mean ${\bf L}$--$M$ interaction on $Y$, under which the $\text{NIE}^{\mathrm{R}}$ satisfies the indirect effect measure criteria in the presence of an exposure-induced confounder.

\section{Indirect effects in the absence of cross-world counterfactual independencies}
\label{sec:no-recanting}
Recall that under Assumptions \ref{assn:NUCA-AY}--\ref{assn:NUCA-AM}, \ref{assn:consistency}, and \ref{assn:positivity}, 
the NIE$^{\text{R}}$ is nonparametrically identified by $\Psi^{\text{NIE}}(P)$. While the absence of an exposure-induced confounder is not necessary for the identification of the NIE$^{\text{R}}$, we have the following theorem regarding the indirect effect measure criteria in the absence of cross-world counterfactual independencies.
\begin{theorem}
    \label{thm:FFRCISTG-SNC}
    Suppose the counterfactual distribution follows the FFRCISTG corresponding to the DAG in Figure \ref{fig:DAG1}. Then the $\text{NIE}^{\text{R}}$ does not satisfy any of the indirect effect measure criteria.
\end{theorem}
The intuition behind the proof comes from the following inequation under the sharp mediational null:
\begin{align*}
    E\left[Y\{a,G(a)\}\right]
    &= E\left[Y\{a,M(a)\}\right] = E\left[Y\{a,M(a^*)\}\right]\\
    &= E\left[\int_m E\left\{Y(a,m)\mid M(a^*)=m, {\bf C}\right\}f_{M(a^*)\mid {\bf C}}(m\mid {\bf C})d\mu(m)\right]\\
    &\neq E\left[\int_m E\left\{Y(a,m)\mid {\bf C}\right\}f_{M(a^*)\mid {\bf C}}(m\mid {\bf C})d\mu(m)\right] = E\left[Y\{a,G(a^*)\}\right].
\end{align*}
in general, since $Y(a,m)\ci M(a^*)\mid {\bf C}$ is not implied by the FFRCISTG. 
However, the sharper mediational null does place constraints on the joint distribution of $Y(a,m)$ and $M(a^*)$, so a fully rigorous proof requires a counterexample to show that the sharper mediational null in addition to the FFRCISTG independence assumptions do not imply this cross-world counterfactual independence. Such a counterexample is provided in the proof in the supporting web materials.

\cite{robins1992identifiability} and \cite{robins2003semantics} showed that the NIE is identified under the FFRCISTG 
when there is no individual-level causal interaction between $A$ and $M$ on $Y$. 
The following theorem relaxes this assumption slightly to one of no interaction on a conditional mean scale. 
\begin{theorem}
    \label{thm:no-AM-interaction}
    Suppose there is no mean 
    causal interaction between $A$ and $M$ on $Y$ within all strata of $M(a^*)$ and $\bf C$ on the additive scale, i.e., 
    \[E\{Y(a,m')-Y(a,m'')-Y(a^*,m')+Y(a^*,m'')\mid M(a^*),{\bf C}\}=0\]
    almost surely for all $m'$ and $m''$. Then under the FFRCISTG corresponding to the DAG in Figure \ref{fig:DAG1}, the NIE is equivalent to the NIE$^{R}$ and the PE$(m)$ for all $m$, and is nonparametrically identified by $\Psi^{\text{NIE}}(P)$.
\end{theorem}
While this no mean causal interaction assumption superficially resembles the identifying assumption that
\[E\{Y(a,m)-Y(a^*,m)\mid M(a^*)=m,{\bf C}\}=E\{Y(a,m)-Y(a^*,m)\mid {\bf C}\}\]
for all $m$ of \cite{petersen2006estimation} and \cite{van2008direct}, these are in fact distinct. \cite{vanderweele2009conceptual} note that the latter assumption is essentially a disjunction of Assumption \ref{assn:NPSEM} and the assumption of no individual-level causal interaction between $A$ and $M$ on $Y$.
\begin{corollary}
    \label{cor:no-AM-interaction}
    The NIE$^{\text{R}}$ and the PE$(m)$ satisfy the indirect effect measure criteria under the conditions in Theorem \ref{thm:no-AM-interaction}.
\end{corollary}

In fact, the no mean causal interaction assumption of Theorem \ref{thm:no-AM-interaction} also identifies the NIE under the FFRCISTG corresponding to the DAG in Figure \ref{fig:DAG2}, i.e., in the presence of an exposure-induced confounder.
\begin{theorem}
    \label{thm:no-AM-interaction-recanting-witness}
    Suppose there is no mean 
    causal interaction between $A$ and $M$ on $Y$ within all strata of $M(a^*)$ and $\bf C$ on the additive scale as in Theorem \ref{thm:no-AM-interaction}. 
    Then under the FFRCISTG corresponding to the DAG in Figure \ref{fig:DAG2}, the NIE is equivalent to the NIE$^{R}$ and the PE$(m)$ for all $m$, and is nonparametrically identified by $\Psi^{\text{NIE}^{\text{R}}}_{L}(P)$.
\end{theorem}
\begin{corollary}
    \label{cor:no-AM-interaction-recanting-witness}
    The NIE$^{\text{R}}$ and the PE$(m)$ satisfy the indirect effect measure criteria under the conditions in Theorem \ref{thm:no-AM-interaction-recanting-witness}.
\end{corollary}

Theorems \ref{thm:no-AM-interaction} and \ref{thm:no-AM-interaction-recanting-witness} highlight the key role $A$--$M$ interaction plays in mediation. Under the no mean causal interaction assumption of Theorems \ref{thm:no-AM-interaction} and \ref{thm:no-AM-interaction-recanting-witness}, the mean reference interaction (definition provided in the proof) of the four-way decomposition of \cite{vanderweele2014unification} 
is zero. 
The CDE$(m)$ and the NIE account for the remaining terms in this decomposition, 
hence, we have $\text{NIE} = \text{TE} - \text{CDE}(m) = \text{PE}(m)$ for all $m$. We have previously seen the portion eliminated to be identified under assumptions that hold both in the presence of an exposure-induced confounder and in the absence of cross-world counterfactual independencies. While we have also seen that the portion eliminated cannot distinguish between mediation and interaction and therefore does not satisfy the sharp(er) null criterion, by ruling out non-mediated interaction, we can indeed rely on it to detect mediation. Likewise, Corollaries \ref{cor:no-AM-interaction} and \ref{cor:no-AM-interaction-recanting-witness} highlight that the same can be said of the NIE$^{\text{R}}$ when non-mediated interaction is ruled out. Indeed, the NIE, NIE$^{\text{R}}$, and PE$(m)$ all align under this no mean causal interaction assumption.


Section S3.2 of the supporting web materials gives additional settings under which the $\text{NIE}^{\mathrm{R}}$ satisfies the indirect effect measure criteria in the absence of cross-world counterfactual independencies. In particular, these criteria are satisfied under the separable effects model of \cite{robins2010alternative}, and the sharp null and sharper null criteria are satisfied when $M$ affects $Y$ for every individual.

\section{Non-mediational interpretations of the NIE$^{\text{R}}$}
\label{sec:interpretations}

While the previous sections have demonstrated that the NIE$^{\text{R}}$ lacks a true indirect effect interpretation, this does not mean that it lacks a useful causal interpretation. In fact, the NIE$^{\text{R}}$ has a natural interpretation as a contrast of two hypothetical stochastic joint interventions \citep{didelez2006direct,geneletti2007identifying}. These can be viewed as stochastic two-stage dynamic treatment regimes where the first stage is an intervention assigning the exposure level and the second stage is an intervention assigning the mediator level. Dynamic treatment regimes have traditionally depended on at most an individual's observed data history prior to each stage, e.g., the intervention assigning the exposure may depend on baseline covariates ${\bf C}$, say $A_{g^*}\sim g_A^*(a\mid {\bf C})$, and the intervention assigning the mediator may depend on both ${\bf C}$ and the intervened exposure $A_{g^*}$, say $M_{g^*}\sim g_M^*(m\mid A_{g^*}, {\bf C})$, where $g_A^*$ and $g_M^*$ are user-specified, potentially degenerate intervention distributions. 

More recently, interventions at a given stage involving the exposure at that same stage have gained popularity \citep{taubman2009intervening, munoz2012population, moore2012causal, haneuse2013estimation, 
young2019inverse, sani2020identification}, and have proven useful for relaxing positivity conditions \citep{kennedy2019nonparametric, papadogeorgou2019causal}. 
Such interventions may depend on the natural value of the observed exposure itself or its conditional distribution given the individual's exposure and covariate history. The joint interventions that are contrasted in the NIE and NIE$^{\text{R}}$ correspond to these two types of interventions depending on the intervention variable in the second stage (i.e., the mediator), respectively. One important distinction, however, is that these two effects involve interventions depending on a counterfactual version of the mediator rather than the observed mediator. When the intervention is a function of the current observed exposure and possibly the covariate and exposure history, the average counterfactual outcome will be the same as that under an intervention setting the exposure to a random draw from the conditional distribution of this function of the observed exposure given the covariate and exposure history. This will not necessarily be the case when considering interventions depending on the counterfactual exposure and its corresponding conditional distribution, as we can see when contrasting the NIE and NIE$^{\text{R}}$.

The joint interventions that are contrasted in the NIE are as follows: one sets $A$ to the level $a$ and $M$ to the natural level the counterfactual $M(a)$ would have taken, and the other sets $A$ to the level $a$ and $M$ to the natural level the counterfactual $M(a^*)$ would have taken. On the other hand, the joint interventions that are contrasted in the NIE$^{\text{R}}$ are as follows: one sets $A$ to the level $a$ and $M$ to a random draw from the conditional distribution of $M(a)$ given ${\bf C}$ (i.e., the realization of the random variable $G(a)$ for a given individual), and the other sets $A$ to the level $a$ and $M$ to a random draw from the conditional distribution of $M(a^*)$ given ${\bf C}$ (i.e., the realization of the random variable $G(a^*)$ for a given individual). 
While these interventions depending on the counterfactual version of the mediator clearly differ from an intervention depending on the observed mediator, under Assumption \ref{assn:NUCA-AM}, the distribution of the former is identified by the conditional distribution of $M$ given ${\bf C}$ and $A$. Thus, the interventions in the NIE$^{\text{R}}$ can alternatively be interpreted as depending on the observed conditional distribution of the mediator. However, the same cannot be said of the NIE, since the counterfactual random variables $M(a)$ and $M(a^*)$ themselves are not observed, and the interventions in the NIE involve setting the mediator to these precise values rather than samples from their conditional distributions.

There is an alternative interpretation of the NIE$^{\text{R}}$ relating to the discussion of the identification formula of the natural direct effect in \cite{van2008direct}. They discuss (a conditional version of) the causal parameter 
\[E\left[\int_m\left\{Y(a,m)-Y(a^*,m)\right\}dF_{M(a^*)\mid {\bf C}}(m\mid {\bf C})\right],\]
which is nonparametrically identified by $\Psi^{\text{NDE}}(P)$ even when Assumption \ref{assn:NPSEM} does not hold. They interpret this effect as ``the [...] expectation [...] of a subject-specific average [...] of the [$m$]-specific individual controlled direct effects [$Y(a,m)-Y(a^*,m)$], averaged with respect to the conditional distribution of [$M(a^*)$] given [$\bf C$].'' This parameter is equivalent to the $\text{NDE}^{R}$. Similarly, we have 
\[\text{NIE}^{\text{R}}=E\left[\int_mY(a,m)\left\{dF_{M(a)\mid {\bf C}}(m\mid {\bf C})-dF_{M(a^*)\mid {\bf C}}(m\mid {\bf C})\right\}\right].\]
Thus, the NIE$^{\text{R}}$ can be interpreted as the difference in means of two subject-specific averages of the $m$-specific individual counterfactual $Y(a,m)$, comparing one that is averaged with respect to the conditional distribution of $M(a)$ given $\bf C$ and another that is averaged with respect to the conditional distribution of $M(a^*)$ given $\bf C$.

\section{Related randomized interventional indirect effect measures}
\label{sec:related}

\cite{lin2017interventional} defined generalizations of the NIE$^{\text{R}}$ to arbitrary path-specific effects, which they showed to be nonparametrically identified under FFRCISTGs corresponding to elaborated versions of the DAG in Figure \ref{fig:DAG2} 
with additional intermediate variables. As these contain the NIE$^{\text{R}}$ as a special case, we already know that one member of this class fails to satisfy the indirect effect measure criteria without stronger assumptions. There exist conventional (i.e., non-randomized interventional) path-specific effects that are nonparametrically identified in this setting, though many are not identified \citep{avin2005identifiability,shpitser2013counterfactual,zhou2021semiparametric}. It seems likely that when the corresponding non-interventional path-specific effect is identified, the interventional and non-interventional path-specific effects will share the same identification formula, and that the former will satisfy the analogous path-specific effect measure criteria with respect to the causal path it is defined in terms of. However, when the corresponding path-specific effect is not identified, it seems unlikely that such an effect will satisfy its corresponding path-specific effect measure criteria. This is merely conjecture, and I do not prove such a result here. 
However, given that one member of this class fails to satisfy its corresponding effect measure criteria when its corresponding non-interventional effect is not identified, the burden of proof would seem to be on demonstrating that any of the other randomized interventional path-specific effects do satisfy them, rather than that they do not.

Variations of randomized interventional indirect effects have also been used in health disparities research, where alternative formulations are of interest due to disagreements around manipulability of certain characteristics such as race. \cite{vanderweele2014causal} first considered a version of the randomized interventional indirect effect that conditions on the exposure being set to a particular level rather than defining the counterfactual outcome under an intervention on the exposure. Instead, the counterfactual is only defined with respect to an intervention on the intermediate variable $M$, and randomized interventions on $M$ corresponding to different levels of the exposure are contrasted. In particular, they define their indirect effect to be 
\[E[Y\{H(a)\}\mid A=a, H(a^*), {\bf C}] - E[Y\{H(a^*)\}\mid A=a, H(a^*), {\bf C}],\]
where $H(a')$ is defined to be a random draw from the conditional distribution of $M$ given $A=a'$ and ${\bf C}$. When the exposure is considered to be manipulable such that counterfactuals $Y(a', m)$ are well-defined and Assumption \ref{assn:NUCA-AM} holds, $H(a')$ is equivalent to $G(a')$; however, this effect definition relaxes both of these assumptions. If additionally, Assumption \ref{assn:NUCA-AY} holds, then the above indirect effect is equivalent to a conditional version of the NIE$^{\text{R}}$, and its expectation will be equivalent to the NIE$^{\text{R}}$ itself. Clearly, under the nonparametric identification assumptions for the NIE$^{\text{R}}$, the indirect effect of \cite{vanderweele2014causal} will not satisfy any of the indirect effect measure criteria since it is equivalent to the NIE$^{\text{R}}$ under these assumptions. \cite{vanderweele2014causal} show their indirect effect to be nonparametrically identified under Assumption \ref{assn:NUCA-MY}, which is clearly weaker than the identifying assumptions of the NIE$^{\text{R}}$, hence it will not satisfy any of the indirect effect measure criteria under this weaker assumption. 

Having said this, \cite{vanderweele2017mediation} made the case that a true mediational interpretation is arguably not the effect measure of primary interest in disparities research. Indeed, if the exposure is thought not to be manipulable, then the indirect effect measure criteria themselves are not meaningful. Instead, this effect is interpreted as the amount of change in the outcome among the population with $A=a$ (e.g., a disadvantaged population) and baseline covariates $\bf C$ if the mediator distribution for this population were to be intervened on to be (randomly at the individual level) set equal to that of the population with $A=a^*$ (e.g., an advantaged population). This interpretation makes the effect useful for identifying sources of disparities between groups that may or may not happen to be causally downstream from group membership. The presence of such an effect suggests that interventions targeting the intermediate variable may be effective at reducing the disparity in view. \cite{jackson2018decomposition} connected the effect of \cite{vanderweele2014causal} to a component of the Kitagana--Blinder--Oaxaca decomposition \citep{kitagawa1955components, blinder1973wage, oaxaca1973male} that is used to study discrimination in the economics literature, and established sets of conditions under which the latter can have a causal interpretation. \cite{jackson2020meaningful} observed that in the bioethics and health services literature, disparities are defined using associational rather than causal language. He instead distinguished sets of ``allowable'' from ``non-allowable'' covariates to consider conditioning on in defining the estimand of interest. Such distinctions are made based on whether one would deem a covariate as an acceptable source of association between group membership and the outcome in the sense of reflecting equity value judgments. Based on this, he presented a refined version of the indirect effect of \cite{vanderweele2014causal}, which considers more meaningful choices of what to condition on in the estimand. The covariates in the conditional expectation of $Y$ and in the intervention distribution $H$ need not align, and each may be distinct from the set of covariates satisfying Assumption \ref{assn:NUCA-MY}. 

\section{Interpretation of the NIE$^\text{R}$ in the Harvard PEPFAR data set}
\label{sec:pepfar}
Recall the Harvard PEPFAR data application introduced in Section \ref{sec:prelims}. The variables $\bf L$ and $\bf C$ have not yet been defined in the context of this example, and $M$ requires further elaboration; the variables $A$ and $Y$ are exactly as described in Section \ref{sec:prelims}. The potential exposure-induced confounder $\bf L$ is a vector consisting of (a) an indicator of the presence of any lab toxicity over the first six months after initiation of therapy, and (b) a trichotomization of a continuous measure of adherence based on pill counts averaged over the same six months, with cutoffs at 80\% and 95\%. The mediator $M$ is the same measure of adherence as in $\bf L$, but averaged over the subsequent six months. 
Lastly, the vector of baseline covariates $\bf C$ consists of sex, age, marital status, WHO stage, hepatitis C virus, hepatitis B virus, CD4+ cell count, viral load, the tertiary hospital affiliated with the patient’s clinic, and whether the patient visited that tertiary hospital or an affiliated clinic.

\cite{miles2017partial} gave both point estimates and partial identification bounds of the natural indirect effect under various sets of causal assumptions. 
Under the NPSEM-IE corresponding to the DAG in Figure \ref{fig:DAG1}, i.e., assuming neither toxicity nor any other event is an exposure-induced confounder, the natural indirect effect was estimated to be $\hat{\psi}^{\text{NIE}}\equiv \hat{\Psi}^{\text{NIE}}(\hat{P})=0.0014$ and was not found to be statistically significant. Nonetheless, for purposes of illustration, I will consider here the causal interpretation of effects corresponding to the identification formula $\Psi^{\text{NIE}}(P)$ if this were hypothetically known to be the true value of this statistical parameter.

Under the FFRCISTG corresponding to the DAG in Figure \ref{fig:DAG2}, the NIE$^{\text{R}}$ is nonparametrically identified, and if one is further willing to assume that there is no exposure-induced confounder, then it is equal to $\Psi^{\text{NIE}}(P)$. According to Theorems \ref{thm:recanting-SNC} and \ref{thm:FFRCISTG-SNC}, the NIE$^{\text{R}}$ lacks a mediational interpretation without any further assumptions. Thus, if the NIE$^{\text{R}}$ were known to be 0.0014, this would not constitute evidence that adherence mediates the effect of ART regimen on virologic failure for any patient. However, this effect does have non-mediational interpretations as discussed in Section \ref{sec:interpretations}. In particular, under the FFRCISTG corresponding to the DAG in Figure \ref{fig:DAG1}, the effect estimate $\hat{\psi}^{\text{NIE}}=0.0014$ has the following interpretation: Under an intervention assigning all patients to receive AZT+3TC+NVP, the risk of virologic failure would be 0.14\% higher under a subsequent intervention forcing adherence to be a random draw from the conditional distribution of the level at which patients would have adhered under AZT+3TC+NVP within the stratum of their observed baseline covariates than it would have been under a subsequent intervention forcing adherence to instead be a random draw from the conditional distribution of the level at which patients would have adhered under TDF+3TC/FTC+NVP within the stratum of their observed baseline covariates. Alternatively, the interpretation of \cite{van2008direct} is as follows: The risk of virologic failure 
under an intervention assigning all patients to receive AZT+3TC+NVP and forcing all of their adherence levels to $m$ is 0.14\% higher when averaging over $m$ according to the distribution of the level at which patients would have adhered under AZT+3TC+NVP within the stratum of their observed baseline covariates 
than it is when averaging over $m$ according to the conditional distribution of the level at which patients would have adhered under TDF+3TC/FTC+NVP within the stratum of their observed baseline covariates.

As has been shown, there are a number of conditions under which the NIE$^{\text{R}}$ 
does satisfy the sharp(er) null criterion. This means that under any such conditions, if the NIE$^{\text{R}}$ were known to be 0.0014, we could conclude that adherence mediates the effect of ART regimen on virologic failure for at least some patients. Here, the meaning of the word ``mediates'' changes slightly depending on which criterion is being considered. The sharp null criterion corresponds to mediation meaning an intervention changing ART regimen leads to a change in adherence level, and this induced change in adherence leads to a change in virologic failure, whereas the sharper null criterion corresponds to mediation meaning an intervention changing ART regimen leads to a change in adherence level, and an intervention changing adherence level leads to a change in virologic failure. Under most, but not all, of the conditions under which the NIE$^{\text{R}}$ satisfies the sharp(er) null criterion, it will also satisfy the monotonicity criterion. In these cases, if the NIE$^{\text{R}}$ were known to be 0.0014, we could conclude that AZT+3TC+NVP causes an increased risk of virologic failure relative to TDF+3TC/FTC+NVP through its effect on adherence for at least some patients. 

The conditions of Theorem \ref{thm:recanting-SNC} imply that 
adherence over the second six months 
is effectively randomized within strata of baseline covariates, ART regimen, levels of toxicity, and adherence over the first six months, such that it is independent of virologic failure under an intervention setting ART regimen and adherence over the second six months to any arbitrary value. 
The conditions of Theorem \ref{thm:FFRCISTG-SNC} imply that 
there is no confounder of the effect of adherence on virologic failure that is differentially affected between the two ART regimens. In particular, toxicity (a) is not differentially affected by the ART regimen, (b) does not affect adherence, or (c) does not affect virologic failure. The potential violation of these three conditions is what motivated the work in \cite{miles2017quantifying}, which deals with toxicity as an exposure-induced confounder. See this article for discussion on why these are thought to be violated. Under either set of conditions for these two theorems, the NIE$^{\text{R}}$ fails to satisfy any of the indirect effect measure criteria, and only has the non-mediational interpretation given above.

The conditions of Theorem \ref{thm:overlap} imply that 
the existence of some patients for whom ART regimen differentially affects adherence and some for whom adherence affects virologic failure means that these two groups overlap in the sense that for some of these patients, \emph{both} their ART regimen affects their adherence level, \emph{and} the resulting change in their adherence level affects whether they experience virologic failure. Under these conditions, the NIE$^{\text{R}}$ will satisfy the sharp(er) null criterion, but not the monotonicity criterion. 

Theorem \ref{thm:no-AM-interaction} and Corollary \ref{cor:no-AM-interaction} will hold if, in addition to the conditions in Theorem \ref{thm:FFRCISTG-SNC}, there is no mean causal interaction between ART regimen and adherence on virologic failure within each stratum of covariates and adherence level under an intervention assigning patients to TDF+3TC/FTC+NVP. Similarly, Theorem \ref{thm:no-AM-interaction-recanting-witness} and Corollary \ref{cor:no-AM-interaction-recanting-witness} will hold if, in addition to the conditions in Theorem \ref{thm:recanting-SNC}, there is no mean causal interaction between ART regimen and adherence on virologic failure within each stratum of covariates and adherence level under an intervention assigning patients to TDF+3TC/FTC+NVP. Under these conditions, the NIE$^{\text{R}}$ will satisfy all of the indirect effect measure criteria. Additionally, the portion eliminated would also be equal to the NIE$^{\text{R}}$ for all levels $m$ of adherence over the second six months. This means that, under the conditions of Theorem \ref{thm:no-AM-interaction}, the differential effect of ART regimen on risk of virologic failure would be reduced by 0.0014 under an intervention setting adherence over the second six months to any particular level. 

Under the conditions of Theorem \ref{thm:no-AM-interaction} 
both the NIE and the NIE$^{\text{R}}$ will be equal to $\Psi^{\text{NIE}}(P)$. As such, we have the following interpretation of the estimate $\hat{\psi}^{\text{NIE}}=0.0014$: The risk of virologic failure is 0.14\% higher under an intervention assigning patients to AZT+3TC+NVP than it would be if adherence were to be forced to the level it would have been had patients instead been assigned to TDF+3TC/FTC+NVP. 
That is, when fixing the ART regimen to AZT+3TC+NVP, it is the change in the risk of virologic failure due to the change in adherence caused by a change in ART regimen from TDF+3TC/FTC+NVP to AZT+3TC+NVP. This is a true mediational interpretation, in that it captures a change in the risk of virologic failure due to a change in adherence induced by a change in ART regimen. 
The interpretation under the conditions of Theorem \ref{thm:no-AM-interaction-recanting-witness} is exactly the same, except the NIE and the NIE$^{\text{R}}$ are instead identified by $\Psi_{L}^{\text{NIE}^{\text{R}}}(P)$.



\section{Discussion}
\label{sec:discussion}
In this article, I have defined the sharp null, sharper null, and monotonicity criteria. I argue that the first is an essential criterion that any effect measure claiming a conventional mediational interpretation with respect to a single mediator or single set of mediators ought to satisfy. The second criterion is a somewhat weaker version not involving so-called cross-world counterfactuals, which some may prefer. In particular, these criteria assert that any true mediational effect measure ought to take its null value when there is no mediated effect for any individual in the population of interest, with each version corresponding to a slightly different interpretation of the term ``mediated effect''. The monotonicity criterion is essential for an indirect effect measure to have utility beyond merely detecting the presence of an indirect effect, and to instead be able to correctly identify, at least for some subjects, the direction of the indirect effect.

Causal mediation analysis is known to be challenging and to rely on a much heavier set of assumptions for identification than most other sorts of causal estimands. The recent dominating trend in methodological development for causal mediation has been to circumvent such assumptions by redefining the target causal parameter to one that closely resembles the natural indirect effect or related path-specific effect, but remains identified even in the presence of exposure-induced confounders and/or in the absence of cross-world counterfactual independence assumptions. This article demonstrates that thus far, these relaxations have resulted in a failure to satisfy the indirect effect measure criteria, and thereby a loss of a true mediational interpretation, at least in terms of the conventional understanding of the word. 
Thus, in the current state of affairs, identification of true causal mediated effects remains challenging. An essential open question is whether there exist any nontrivial effect measures that satisfy these criteria, either in the presence of an exposure-induced confounder or in the absence of cross-world counterfactual independencies. The existence of such an effect measure would have profound implications for longitudinal mediation with time-varying exposures--a setting in which exposure-induced confounders proliferate.

Despite these negative results for randomized interventional indirect effects, I have reviewed alternative meaningful causal interpretations of such effects. These involve stochastic and dynamic interventions on the mediator that depend on the distribution of either the counterfactual or observed value of the mediator. In practice, there are many possible stochastic and dynamic interventions that can be performed on the mediator, and care should be taken to decide which of these interventions is of greatest scientific interest for a given problem. This will at times be the interventions assigning $M$ to follow the conditional distributions of $M(a)$ or $M(a^*)$ given baseline covariates (as in the NIE$^{\mathrm{R}}$), or alternatively the conditional distribution of $M$ given $A=a$ or $A=a^*$ and baseline covariates (as in the randomized interventional effects that avoid interventions on the exposure described in Section \ref{sec:related}). In fact, we have seen the latter interventions to be directly relevant to studies seeking to identify and eliminate sources of disparities. In such settings, quantities related to randomized interventional indirect effects have extremely useful interpretations apart from mediational interpretations. However, based on the findings in this article, one should resist the temptation to imbue randomized interventional indirect effects with a mechanistic mediational interpretation when the indirect effect measure criteria are 
not satisfied, 
or to select such an effect as a substitute for the natural indirect effect if the scientific question is focused on mediation.


\section*{\centering Acknowledgments}
I would like to thank John Jackson, Ilya Shpitser, and Eric Tchetgen Tchetgen for helpful discussions, and two reviewers and an associate editor for constructive comments that have helped to strengthen the content and clarity of the article. I would also like to thank Phyllis Kanki (Harvard) and APIN Public Health Initiatives for the use of the PEPFAR data example.

\bibliographystyle{apalike}

\bibliography{references}

\end{document}


\def\spacingset#1{\renewcommand{\baselinestretch}%
{#1}\small\normalsize} \spacingset{1}


\if1\blind { \title{Supporting materials}
\date{}
    \maketitle
} \fi

\if0\blind
{
  \bigskip
  \bigskip
  \bigskip
  \begin{center}
    {\LARGE\bf Supporting materials for \emph{On the Causal Interpretation of Randomized Interventional Indirect Effects}}
\end{center}
  \medskip
} \fi

\spacingset{1.5} 

\section{General definitions of NPSEM-IE and FFRCISTG}
For a given causal DAG $\mathcal{G}$, let $\boldsymbol{\mathcal{V}}$ denote the set of nodes in $\mathcal{G}$, $\mathbf{pa}_{V}$ denote the set of parents of $V$ in $\mathcal{G}$ for each $V\in\boldsymbol{\mathcal{V}}$, and ${\bf v}_{\bf X}$ denote the subset of $\bf v\in \text{supp}(\boldsymbol{\mathcal{V}})$ corresponding to the subset $\bf X\subset \boldsymbol{\mathcal{V}}$, where $\text{supp}(\cdot)$ indicates the support of its argument. Further, let $V({\mathbf{pa}}_{V})$ denote the counterfactual value of $V$ under an intervention assigning its parents in $\mathcal{G}$ to ${\mathbf{pa}}_{V}$. The first causal model is the nonparametric structural equation model with independent errors (NPSEM-IE) \citep{pearl1995causal}, 
which consists of all counterfactual probability distributions for which all sets of counterfactuals in the collection $\{\{V({\mathbf{pa}}_{V})\mid \forall \; {\mathbf{pa}}_{V}\}\mid \forall \; V\in\boldsymbol{\mathcal{V}}\}$ are mutually independent. The second is the finest fully randomized causally interpretable structured tree graph (FFRCISTG) \citep{robins1986new}, which consists of all counterfactual probability distributions for which all counterfactuals in the collection $\{V({\mathbf{pa}}_{V})\mid \forall \; V\in\boldsymbol{\mathcal{V}}, {\mathbf{pa}}_{V}={\bf v}_{{\mathbf{pa}}_{V}}\}$ are mutually independent for each ${\bf v}\in\text{supp}(\boldsymbol{\mathcal{V}})$. The independence assumptions imposed by the NPSEM-IE imply those of the FFRCISTG, hence the latter model contains in the former. The NPSEM-IE includes counterfactual independencies that are known as \emph{cross-world counterfactual} independencies.

\section{Results for randomized interventional indirect effects defined in terms of an exposure-induced confounder $\bf L$}
\begin{theorem}
    \label{thm:NIE-RL}
    Under the NPSEM-IE corresponding to the DAG in Figure 2 and Assumption 6, the effect measure $E[Y\{a,G(a\mid {\bf C}, {\bf L})\}]-E[Y\{a,G(a^*\mid {\bf C}, {\bf L})\}]$ is nonparametrically identified by
    \[E\left\{E(Y\mid {\bf L}, a, {\bf C})\right\}-E\left[E\left\{E(Y\mid M, {\bf L}, a, {\bf C})\mid {\bf L}, a^*, {\bf C}\right\}\right],\]
    but does not satisfy any of the indirect effect measure criteria.
\end{theorem}
The proof of Theorem \ref{thm:NIE-RL} uses the same counterexample as that used in the proof of Theorem 1. Thus, the role that $\bf L$ plays in Theorem \ref{thm:NIE-RL} is identical to that in Theorem 1. The former theorem simply demonstrates that satisfaction of the indirect effect measure criteria is not recovered by further conditioning on $\bf L$ in the construction of the randomized analog to the counterfactual $M(a')$.
\begin{theorem}
    \label{thm:NIE-RLa}
    Under the NPSEM-IE corresponding to the DAG in Figure 2 and Assumption 6, the effect measure $E(Y[a,G\{a\mid {\bf C}, {\bf L}(a)\}])-E(Y[a,G\{a^*\mid {\bf C}, {\bf L}(a^*)\}])$ is equal to the NIE, and is not nonparametrically identified.
\end{theorem}
\cite{zheng2017longitudinal} (in the point treatment special case of their longitudinal setting) and \cite{nguyen2022clarifying} consider an effect that is closely related to the effect considered in Theorem \ref{thm:NIE-RLa}, but with a subtle distinction.  Their effect is also defined with respect to the same interventional distribution, i.e., $M(a')\mid {\bf C, L}(a')$; however, instead of drawing from this distribution at the value naturally taken by the counterfactual ${\bf L}(a')$, they draw from the distribution $M(a')\mid {\bf C, L}(a')=\boldsymbol{\ell}$, where $\boldsymbol\ell$ is the natural value taken instead by ${\bf L}(a)$. The nonparametric identification formula for this effect is in fact equivalent to that of the path-specific effect through $M$ but not $\bf L$ considered in \cite{avin2005identifiability}, \cite{miles2017quantifying}, and \cite{miles2020semiparametric} for an appropriate choice of the treatment levels in the effect definition. As this path-specific effect only captures part of the indirect effect of $A$ on $Y$ through $M$, it seems quite likely that there could be distributions in which there is cancellation of effects between effects along the paths $A\rightarrow M\rightarrow Y$ and $A\rightarrow {\bf L} \rightarrow M\rightarrow Y$ at the individual level, such that the sharp(er) mediational null with respect to the indirect effect along all paths from $A$ to $Y$ through $M$ holds. However, in this case, the path-specific effect may be non-null. Since the effect of \cite{zheng2017longitudinal} and \cite{nguyen2022clarifying} is equal to this path-specific effect when both are identified, the sharp(er) null criterion appears unlikely to hold for either of these effects. Having said this, the path-specific effect does have a meaningful mediational interpretation, albeit with respect to a different causal pathway than the one quantified by the $\text{NIE}$, and the randomized interventional effect of \cite{zheng2017longitudinal} and \cite{nguyen2022clarifying} offer an alternative interventional interpretation of this effect.

\section{Additional assumptions under which the NIE$^{\text{R}}$ satisfies the indirect effect measure criteria}
\subsection{In the presence of an exposure-induced confounder}

An assumption under which satisfaction of the sharp null, sharper null, and monotonicity criteria is recovered is that there is no mean ${\bf L}$--$M$ interaction on $Y$ on the additive scale, i.e., 
\[E(Y\mid m', {\boldsymbol \ell'}, a', {\bf C}) - E(Y\mid m', {\boldsymbol \ell''}, a', {\bf C}) = E(Y\mid m'', {\boldsymbol \ell'}, a', {\bf C}) - E(Y\mid m'', {\boldsymbol \ell''}, a', {\bf C})\]
almost surely for all $a'$, ${\boldsymbol \ell'}$, ${\boldsymbol \ell''}$, $m'$, and $m''$. \cite{tchetgen2014identification} showed that the NIE is nonparametrically identified under this assumption even in the presence of an exposure-induced confounder.
\begin{theorem}
    \label{thm:no-LM-interaction}
    Suppose there is no mean interaction between ${\bf L}$ and $M$ on $Y$ on the additive scale. Then under the NPSEM-IE corresponding to the DAG in Figure 2 and Assumption 6, $\text{NIE}^{\mathrm{R}}=\text{NIE}$, and both satisfy the indirect effect measure criteria.
\end{theorem}

While satisfaction of these criteria is recovered for the NIE$^{\text{R}}$ under these assumptions, they also suffice to nonparametrically identify the NIE. Thus, under such a setting, the NIE$^{\text{R}}$ does not appear to offer a clear advantage over the NIE in terms of quantifying a mediated effect.

\cite{rudolph2018robust} considered estimation of the NIE$^{\text{R}}$ in the presence of an exposure-induced confounder under an instrumental variable-like setting in which the exposure has no direct effect on the mediator with respect to the exposure-induced confounder. That is, $A$ is playing the role of an instrumental variable with respect to a study of the effect of $\bf L$ on $M$. The counterexample in the proof of Theorem 1 does not belong to this model, as $A$ does have a direct effect on $M$ with respect to $L$, hence it is unclear whether the NIE$^{\text{R}}$ will satisfy the indirect effect measure criteria in such a setting. This is an interesting open question worth exploring.

\subsection{In the absence of cross-world counterfactual independencies}
Another way the indirect effect measure criteria can be recovered is under what are known as separable effects of the exposure. This assumption was first proposed by \cite{robins2010alternative} as a way to identify the NIE without requiring cross-world counterfactual independence assumptions. This has since gained popularity both for identifying mediated effects \citep{didelez2019defining, robins2022interventionist} and for handling competing risks in survival analysis \citep{stensrud2020separable}. The separable effects model is the FFRCISTG corresponding to the DAG in Figure \ref{fig:DAG3}, which is an extended version of the DAG in Figure 1 where the red, bolded arrows indicate a deterministic relationship. 
\begin{figure}[h]
    \centering
    \begin{tikzpicture}[baseline={(A)}, 
    ->, line width=1pt]
    \tikzstyle{every state}=[draw=none]
    \node[shape=circle, draw] (C) at (-2,0) {${\bf C}$};
    \node[shape=circle, draw] (A) at (0,0) {$A$};
    \node[shape=circle, draw] (N) at (2,0) {$N$};
    \node[shape=circle, draw] (O) at (2,2) {$O$};
    \node[shape=circle, draw] (M) at (4,0) {$M$};
    \node[shape=circle, draw] (Y) at (6,0) {$Y$};

      \path     (C) edge (A)
                (C) edge  [bend right=30] (M)
                (C) edge  [bend right=45] (Y)
                (M) edge (Y)
                (A) edge  [color=red, line width=2pt] (N)
                (A) edge  [color=red, line width=2pt] (O)
                (N) edge (M)
                (O) edge (Y)
                          ;
    \end{tikzpicture}
    \caption{The extended mediation DAG with separable effects of $A$ on $M$ and $Y$.\label{fig:DAG3}}
\end{figure}
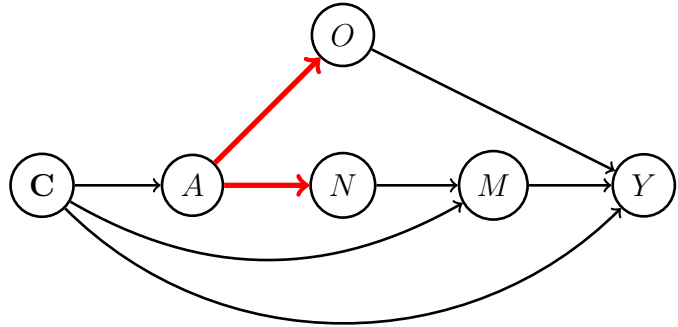
The model asserts that the effect of $A$ on $M$ and $Y$ can be separated into two components, $N$ and $O$, where $N$ completely mediates the effect of $A$ on $M$ and has no direct effect on $Y$, and $O$ completely mediates the direct effect of $A$ on $Y$ and has no effect on $M$. \cite{robins2010alternative} showed that the NIE is nonparametrically identified under the FFRCISTG corresponding to a version of the DAG in Figure \ref{fig:DAG3} where $A$ is randomized, such that there is no arrow from $\bf C$ to $A$. The extension to the observational setting is trivial.
\begin{theorem}
    \label{thm:separable-effects}
    Under the FFRCISTG corresponding to the DAG in Figure \ref{fig:DAG3}, the NIE$^{\text{R}}$ is equal to the NIE and satisfies the indirect effect measure criteria.
\end{theorem}

Another way the sharp null and sharper null criteria can be recovered is if 
$M$ affects $Y$ for every individual. 
\begin{theorem}
    \label{thm:m-always-affects-y}
    Suppose that for each subject $i$, either $Y_i(a^*,m)\neq Y_i(a^*,m')$ or $Y_i(a,m)\neq Y_i(a,m')$ for all $m$ and $m'$ in the support of $M$. Then under the FFRCISTG corresponding to the DAG in Figure 1, the NIE$^{\text{R}}$ satisfies both the sharp null and sharper null criteria.
\end{theorem}
In this case, the NIE is not nonparametrically identified, so the NIE$^{\text{R}}$ does hold an advantage over the NIE as an indirect effect measure in this setting. However, it seems unlikely in most practical applications that one would have knowledge that there is an effect of $M$ on $Y$ for each individual. Additionally, it is not clear in this case whether the NIE$^{\text{R}}$ satisfies the monotonicity criterion.

\section{Proofs}
\begin{proof}[Proof of Theorem 1]
Suppose $A$ is randomized such that ${\bf C}$ is the empty set with $a=1$ and $a^*=0$, 
and suppose the counterfactual distribution belonging to the NPSEM-IE is
\begin{align*}
    A &= \varepsilon_A \\
    L &= A\varepsilon_L + (1-A)(1-\varepsilon_L) \\
    M &= (A+L-AL)\varepsilon_M + (1-A)(1-L)(1-\varepsilon_M) \\
    Y &= (1-A)LM+A(L+M-LM),
\end{align*}
where $\varepsilon_A\sim \text{Bernoulli}(1/2)$, $\varepsilon_L\sim \text{Bernoulli}(\pi)$, $\varepsilon_M\sim \text{Bernoulli}(\beta)$, and $\varepsilon_A$, $\varepsilon_L$, and $\varepsilon_M$ are mutually independent. Thus, when $\varepsilon_L = 0$, 
\[M(a^*)=M\{a^*,L(a^*)\}=M(a^*,1)=\varepsilon_M=M(a,0)=M\{a,L(a)\}=M(a),\] 
so $A$ does not affect $M$, and $Y\{a',M(a)\} = Y\{a',M(a^*)\}$ for both $a'\in\{a^*,a\}$. When $\varepsilon_L = 1$,
\begin{align*}
    Y(a,1) &= Y\{a,L(a),1\} = f_Y(a,a,1) = 1 = f_Y(a,a,0) = Y\{a,L(a),0\} = Y(a,0)
\end{align*}
and
\begin{align*}
    Y(a^*,1) &= Y\{a^*,L(a^*),1\} = f_Y(a^*,a^*,1) = 0\\ 
    &= f_Y(a^*,a^*,0) = Y\{a^*,L(a^*),0\} = Y(a^*,0),
\end{align*}
so $M$ does not affect $Y$. Thus, the sharper mediational null holds 
since $\varepsilon_L$ is either zero or one. However, under this counterfactual distribution, we have $\mathrm{NIE}^{\mathrm{R}}=\pi(1-\pi)(2\beta-1)$, 
which does not equal zero in general. In fact, NIE$^{\mathrm{R}}$ goes to -1/4 as $\pi$ goes to 1/2 and $\beta$ goes to 0, and NIE$^{\mathrm{R}}$ goes to 1/4 as $\pi$ goes to 1/2 and $\beta$ goes to 1, hence NIE$^{\mathrm{R}}$ is not bounded away from -1/4 or 1/4 even when there is no individual-level indirect effect for any subject. Thus, the NIE$^{\mathrm{R}}$ does not satisfy the sharper null criterion, and therefore neither does it satisfy the sharp null or monotonicity criteria.
\end{proof}

\begin{proof}[Proof of Theorem 2]
    The sharp mediational null implies that $Y_{k}\{a,M(a^*)\} = Y_{k}\{a,M(a)\}$ for all $k$, hence either $M_i(a) = M_i(a^*)$ for all $i$ or $Y_j(a,m) = Y_j(a,m')$ for all $j$ and $m\neq m'$. In the first case, we have $f_{M(a)}(m)=f_{M(a^*)}(m)$ for all $m$, which implies $f_{M\mid A, {\bf C}}(m\mid a, {\bf C})=f_{M\mid A, {\bf C}}(m\mid a^*, {\bf C})$ almost surely for all $m$, so 
    \begin{align*}
        \text{NIE}^{\text{R}} =&\; E\bigg[\int_m \int_{\boldsymbol{\ell}} E(Y\mid m,\boldsymbol{\ell},a,{\bf C})dF_{{\bf L}\mid A, {\bf C}}(\boldsymbol{\ell}\mid a, {\bf C})\\
        &\times \{dF_{M\mid A,{\bf C}}(m\mid a,{\bf C}) - dF_{M\mid A,{\bf C}}(m\mid a^*,{\bf C})\}\bigg]\\
        =&\; 0.
    \end{align*}
    In the second case, $f_{Y(a,m)}(y) = f_{Y(a,m')}(y)$ for all $m$, which implies $E(Y\mid m,{\bf L},a,{\bf C}) = E(Y\mid {\bf L},a,{\bf C})$ almost surely for all $m$, so
    \begin{align*}
        \text{NIE}^{\text{R}} =&\; E\bigg[\int_m \int_{\boldsymbol{\ell}} E(Y\mid m,\boldsymbol{\ell},a,{\bf C})dF_{{\bf L}\mid A, {\bf C}}(\boldsymbol{\ell}\mid a, {\bf C})\\
        &\times \{dF_{M\mid A,{\bf C}}(m\mid a,{\bf C}) - dF_{M\mid A,{\bf C}}(m\mid a^*,{\bf C})\}\bigg]\\
        =&\; E\bigg[\int_{\boldsymbol{\ell}} E(Y\mid \boldsymbol{\ell},a,{\bf C})dF_{{\bf L}\mid A, {\bf C}}(\boldsymbol{\ell}\mid a, {\bf C})\\
        &\times \int_m \{dF_{M\mid A,{\bf C}}(m\mid a,{\bf C}) - dF_{M\mid A,{\bf C}}(m\mid a^*,{\bf C})\}\bigg]\\
        =&\; 0.
    \end{align*}
    Thus, in either case, $\text{NIE}^{\text{R}}=0$, so the $\text{NIE}^{\text{R}}$ satisfies the sharp null criterion, and in turn the sharper null criterion.

    Now consider the following counterfactual distribution belonging to the NPSEM-IE corresponding to the DAG in Figure 2:
    \begin{align*}
        A &= \varepsilon_A \\
        L &= (1-A)I(\varepsilon_L = 0) + AI(\varepsilon_L = 1) + 2I(\varepsilon_L = 2) \\
        M &= \{(A+L-AL)\varepsilon_M + (1-A)(1-L)(1-\varepsilon_M)\}I(L \neq 2) + AI(L = 2) \\
        Y &= \{(1-A)LM+A(L+M-LM)\}I(L \neq 2) + MI(L = 2),
    \end{align*}
    where $\varepsilon_A\sim \text{Bernoulli}(1/2)$, $\varepsilon_L\sim \text{Categorical}(\pi_0,\pi_1,\pi_2)$, $\varepsilon_M\sim \text{Bernoulli}(\beta)$, and $\varepsilon_A$, $\varepsilon_L$, and $\varepsilon_M$ are mutually independent. When $\varepsilon_L\neq 2$, this reduces to the distribution in the proof of Theorem 1, so $Y\{a',M(a^*)\}=Y\{a',M(a)\}$ for each $a'\in\{a^*,a\}$. When $\varepsilon_L=2$, $L(a')=2$, $M(a')=a'$, and $Y(a',m)=Y(a',2,m)=m$ for all $a'\in\{a^*,a\}$ and $m$, so 
    \[Y\{a,M(a)\}=Y(a,2,a)=a > a^* = Y(a,2,a^*) = Y\{a,M(a^*)\},\] 
    and 
    \[Y\{a^*,M(a)\} = Y(a^*,2,a) = a > a^* = Y(a^*,2,a^*) = Y\{a^*,M(a^*)\}.\]
    Thus, mediational monotonicity holds. However, $\text{NIE}^{\text{R}}=\pi_1(1-\pi_1)(2\beta-1)+\pi_2$, which goes to $\pi_2-1/4$ as $\pi_1\rightarrow 1/2$ and $\beta\rightarrow 0$. Thus, when $\pi_2<1/4$, the $\text{NIE}^{\text{R}}$ can be negative, which violates the monotonicity criterion.
\end{proof}

\begin{proof}[Proof of Theorem 3]
    Consider the following counterfactual distribution. Suppose $A\sim \text{Bernoulli}(1/2)$, i.e., $A$ is randomized, $M(a)\sim\text{Bernoulli}(\pi)$, and 
    \begin{align*}
        \{Y(a,0), Y(a,1)\} &= \left\{\begin{array}{ll}
            (0, 0) & \text{with probability } \beta_1 \\
            (0, 1) & \text{with probability } \beta_2 \\
            (1, 0) & \text{with probability } \beta_3 \\
            (1, 1) & \text{with probability } \beta_4 \\
        \end{array}\right.,
    \end{align*}
    which are independent of $M(a)$, where $\beta_1+\beta_2+\beta_3+\beta_4=1$. Suppose further that 
    \[M(a^*)=Y(a,0)Y(a,1) + M(a)\lvert Y(a,1)-Y(a,0)\rvert,\] 
    and
    \begin{align*}
        \{Y(a^*,0), Y(a^*,1)\} &= \left\{\begin{array}{ll}
            (0, 0) & \text{with probability } 1-\gamma \\
            (1, 1) & \text{with probability } \gamma \\
        \end{array}\right.,
    \end{align*}
    where $0<\gamma<1$. This counterfactual distribution belongs to the FFRCISTG model, but not the NPSEM-IE since $M(a^*)$ is independent of neither $Y(a,0)$ nor $Y(a,1)$. 
    
    We have $Y(a^*,0) = Y(a^*,1)$ almost surely, and when $Y(a,0)\neq Y(a,1)$, $M(a^*)=M(a)$. 
    Thus, the sharper mediational null holds. However, under this counterfactual distribution, we have 
    $\text{NIE}^{\text{R}} = \{(1-\pi)\beta_4 - \pi\beta_1\}(\beta_3 - \beta_2)$, which does not equal zero in general. This will go to $1/4$ as $\pi\rightarrow 0$, $\beta_1\rightarrow 0$, $\beta_2\rightarrow 0$, $\beta_3\rightarrow 1/2$, and $\beta_4\rightarrow 1/2$, and will go to $-1/4$ as $\pi\rightarrow 0$, $\beta_1\rightarrow 0$, $\beta_2\rightarrow 1/2$, $\beta_3\rightarrow 0$, and $\beta_4\rightarrow 1/2$. Thus, the $\text{NIE}^{\text{R}}$ fails to satisfy the sharper null criterion, and in turn the sharp null criterion as well.
\end{proof}

\begin{proof}[Proof of Theorem 4]
    Under this model, $E[Y\{a,M(a)\}]$ is known to be identified by $E\{E(Y\mid a,{\bf C})\}$, which equals $E[E\{E(Y\mid M,a,{\bf C})\mid a,{\bf C}\}]$. Observe that no mean interaction given $M(a^*)$ and $\bf C$
    \[E\{Y(a',m')-Y(a',m'')-Y(a'',m')+Y(a'',m'')\mid M(a^*),{\bf C}\}=0\]
    implies the weaker no mean interaction statement given $\bf C$ alone:
    \[E\{Y(a',m')-Y(a',m'')-Y(a'',m')+Y(a'',m'')\mid {\bf C}\}=0.\]
    For a given fixed level $m'$,
    \begin{align*}
        &E\left[Y\left\{a,M(a^*)\right\}\right]\\ 
        =&E\left[\int_m E\left\{Y\left(a,m\right)\mid M(a^*)=m, {\bf C}\right\}dF_{M(a^*)\mid{\bf C}}(m\mid {\bf C})\right]\\ 
        =&E\left[\int_m E\left\{Y(a,m')+Y(a^*,m)-Y(a^*,m')\mid M(a^*)=m, {\bf C}\right\}dF_{M(a^*)\mid{\bf C}}(m\mid {\bf C})\right]\\ 
        =&E\left(E\left[E\left\{Y(a,m')\mid M(a^*), {\bf C}\right\}\mid {\bf C}\right]\right.\\
        &+\int_m E\left\{Y(a^*,m)\mid M(a^*)=m, {\bf C}\right\}dF_{M(a^*)\mid{\bf C}}(m\mid {\bf C})\\
        &\left.-E\left[E\left\{Y(a^*,m')\mid M(a^*), {\bf C}\right\}\mid {\bf C}\right]\right)\\ 
        =&E\left[E\left\{Y(a,m')\mid {\bf C}\right\}+\int_m E\left\{Y(a^*,m)\mid {\bf C}\right\}dF_{M(a^*)\mid{\bf C}}(m\mid {\bf C})-E\left\{Y(a^*,m')\mid {\bf C}\right\}\right]\\ 
        =& E\left[\int_m E\left\{Y(a, m')+Y\left(a^*, m\right)-Y(a^*, m')\mid {\bf C}\right\}dF_{M(a^*)\mid {\bf C}}(m\mid {\bf C})\right]\\
        =& E\left[\int_m E\left\{Y(a, m)\mid {\bf C}\right\}dF_{M(a^*)\mid {\bf C}}(m\mid {\bf C})\right]\\
        =& E\left\{\int_m E\left(Y\mid M=m, A=a, {\bf C}\right)dF_{M\mid A,{\bf C}}(m\mid a^*, {\bf C})\right\}\\
        =& E\left[E\left\{E\left(Y\mid M, A=a, {\bf C}\right)\mid A=a^*, {\bf C}\right\}\right].
    \end{align*}
    Thus, $\text{NIE}=\Psi^{\text{NIE}}(P)$.

    Under the FFRCISTG corresponding to the DAG in Figure 1, the NIE$^{\text{R}}$ is also known to be identified by $=\Psi^{\text{NIE}}(P)$. By marginalizing over the four-way decomposition in \cite{vanderweele2014unification}, $\text{TE} = \text{CDE}(m) + \text{INT}_{\text{ref}}(m,m') + \text{NIE}$, where 
    \[\text{INT}_{\text{ref}}(m,m') = E[\{Y(a,m)-Y(a,m')-Y(a^*,m)+Y(a^*,m')\}M(a^*)]\]
    for all $m$ and $m'$. Under the no mean causal interaction assumption, we have
    \begin{align*}
        \text{INT}_{\text{ref}}(m,m') &= E[\{Y(a,m)-Y(a,m')-Y(a^*,m)+Y(a^*,m')\}M(a^*)] \\
        &= E[E\{Y(a,m)-Y(a,m')-Y(a^*,m)+Y(a^*,m')\mid M(a^*),{\bf C}\}M(a^*)] \\
        &= 0.
    \end{align*}
    Thus, $\text{NIE} = \text{TE} - \text{CDE}(m) = \text{PE}(m)$.
\end{proof}

\begin{proof}[Proof of Corollary 1]
    Since the NIE$^{\text{R}}$ and $\text{PE}(m)$ are also identified by $\Psi^{\text{NIE}}(P)$ and the NIE satisfies the sharp null criterion, the NIE$^{\text{R}}$ and $\text{PE}(m)$ must also satisfy the indirect effect criteria in this setting.
\end{proof}

\begin{proof}[Proof of Theorem 5]
    First, we have
    \begin{align*}
        &E\left[Y\left\{a,M(a)\right\}\right]\\ 
        =&\; E\left[\int_{\boldsymbol{\ell}} \int_m E\left\{Y\left(a,\boldsymbol{\ell},m\right)\mid M(a)=m, {\bf L}(a)=\boldsymbol{\ell}, {\bf C}\right\}dF_{M(a)\mid {\bf L}(a), {\bf C}}(m\mid \boldsymbol{\ell}, {\bf C})dF_{{\bf L}(a)\mid {\bf C}}(\boldsymbol{\ell}\mid {\bf C})\right]\\ 
        =&\;  E\left[\int_{\boldsymbol{\ell}} \int_m E\left\{Y(a, \boldsymbol{\ell}, m)\mid {\bf C}\right\}dF_{M(a)\mid {\bf C}}(m\mid {\bf C})dF_{{\bf L}(a)\mid {\bf C}}(\boldsymbol{\ell}\mid {\bf C})\right]\\
        =&\;  E\left\{\int_{\boldsymbol{\ell}} \int_m E\left(Y\mid m, \boldsymbol{\ell}, a, {\bf C}\right)dF_{M\mid A,{\bf C}}(m\mid a, {\bf C})dF_{{\bf L}\mid A, {\bf C}}(\boldsymbol{\ell}\mid a,{\bf C})\right\}.
    \end{align*}

    From the proof of Theorem 4, the derivation showing
    \[E\left[Y\left\{a,M(a^*)\right\}\right] = E\left[\int_m E\left\{Y(a, m)\mid {\bf C}\right\}dF_{M(a^*)\mid {\bf C}}(m\mid {\bf C})\right]\]
    follows identically under the FFRCISTG corresponding to the DAG in Figure 2; however, $E\left\{Y(a, m)\mid {\bf C}\right\}$ is identified differently due to the presence of $\bf L$. Instead, we have
    \begin{align*}
        &E\left[\int_m E\left\{Y(a, m)\mid {\bf C}\right\}dF_{M(a^*)\mid {\bf C}}(m\mid {\bf C})\right]\\
        =&\; E\left\{\int_m\int_{\boldsymbol{\ell}} E\left(Y\mid m, \boldsymbol{\ell}, a, {\bf C}\right)dF_{{\bf L}\mid A, {\bf C}}(\boldsymbol{\ell}\mid a, {\bf C})dF_{M\mid A, {\bf C}}(m\mid a^*, {\bf C})\right\}.
    \end{align*}
    Thus, $\text{NIE}=\Psi^{\text{NIE}_{L}^{\text{R}}}(P)$.
\end{proof}

\begin{proof}[Proof of Corollary 2]
    Since the NIE$^{\text{R}}$ is also identified by $\Psi^{\text{NIE}_L^{\text{R}}}(P)$ and the NIE satisfies the sharp null criterion, the NIE$^{\text{R}}$ must also satisfy the sharp null criterion in this setting, which in turn implies the NIE$^{\text{R}}$ satisfies the sharper null criterion.
\end{proof}

\begin{proof}[Proof of Theorem \ref{thm:NIE-RL}]
    For each $a'\in\{a^*,a\}$,
    \begin{align*}
        &E\left[Y\left\{a,G(a'\mid {\bf C}, {\bf L})\right\}\right]\\ 
        &= E\left[\int_{\boldsymbol\ell}\int_mE\left\{Y(a,m)\mid G(a'\mid {\boldsymbol\ell},{\bf C})=m,{\boldsymbol\ell},{\bf C}\right\}dF_{G(a'\mid {\bf L}, {\bf C})\mid {\bf L}, {\bf C}}(m\mid {\boldsymbol\ell}, {\bf C})dF_{\bf L\mid C}(\boldsymbol\ell\mid {\bf C})\right]\\
        &= E\left[\int_{\boldsymbol\ell}\int_mE\left\{Y(a,m)\mid {\boldsymbol\ell},{\bf C}\right\}dF_{M(a')\mid {\bf L}, {\bf C}}(m\mid {\boldsymbol\ell}, {\bf C})dF_{\bf L\mid C}(\boldsymbol\ell\mid {\bf C})\right]\\
        &= E\left[\int_{\boldsymbol\ell}\int_mE\left\{Y(a,m)\mid m,{\boldsymbol\ell},a,{\bf C}\right\}dF_{M(a')\mid {\bf L}, A, {\bf C}}(m\mid {\boldsymbol\ell}, a', {\bf C})dF_{\bf L\mid C}(\boldsymbol\ell\mid {\bf C})\right]\\
        &= E\left\{\int_{\boldsymbol\ell}\int_mE\left(Y\mid m,{\boldsymbol\ell},a,{\bf C}\right)dF_{M\mid {\bf L}, A, {\bf C}}(m\mid {\boldsymbol\ell}, a', {\bf C})dF_{\bf L\mid C}(\boldsymbol\ell\mid {\bf C})\right\},
    \end{align*}
    hence
    \begin{align*}
        &E\left[Y\left\{a,G(a\mid {\bf C}, {\bf L})\right\}\right]-E\left[Y\left\{a,G(a^*\mid {\bf C}, {\bf L})\right\}\right]\\
        &= E\left[\int_{\boldsymbol\ell}\int_mE\left(Y\mid m,{\boldsymbol\ell},a,{\bf C}\right)\left\{dF_{M\mid {\bf L}, A, {\bf C}}(m\mid {\boldsymbol\ell}, a, {\bf C})\right.\right.\\
        &\left.\left.-dF_{M\mid {\bf L}, A, {\bf C}}(m\mid {\boldsymbol\ell}, a^*, {\bf C})\right\}dF_{\bf L\mid C}(\boldsymbol\ell\mid {\bf C})\right].
    \end{align*}
    Under the counterfactual distribution in the proof of Theorem 1,
    \[E\left[Y\left\{a,G(a\mid {\bf C}, {\bf L})\right\}\right]-E\left[Y\left\{a,G(a^*\mid {\bf C}, {\bf L})\right\}\right] = \beta - 1/2.\]
    Thus, while the sharp mediational null holds under this counterfactual distribution, this effect measure is not zero in general, and may take any value in $(-1/2,1/2)$ depending on the value of $\beta$.
\end{proof}

\begin{proof}[Proof of Theorem \ref{thm:NIE-RLa}]
    \begin{align*}
        &E\left(Y\left[a,G\left\{a\mid {\bf C}, {\bf L}(a)\right\}\right]\right)\\
        =& E\left(\iint\limits_{{\boldsymbol \ell},m}E\left[Y(a,m)\mid G\{a\mid C,{\bf L}(a)\}=m,{\bf L}(a)={\boldsymbol \ell}, C\right]\right.\\
        &\times \left.dF_{G\{a\mid C,{\bf L}(a)\}\mid {\bf L}(a),C}(m\mid {\boldsymbol \ell},C)dF_{{\bf L}(a)\mid C}({\boldsymbol \ell}\mid C)\right)\\
        =& E\left[\iint\limits_{{\boldsymbol \ell},m}E\left\{Y(a,{\boldsymbol \ell},m)\mid {\bf L}(a)={\boldsymbol \ell}, C\right\}dF_{M(a)\mid {\bf L}(a),C}(m\mid {\boldsymbol \ell},C)dF_{{\bf L}(a)\mid C}({\boldsymbol \ell}\mid C)\right]\\
        =& E\left[\iint\limits_{{\boldsymbol \ell},m}E\left\{Y(a,{\boldsymbol \ell},m)\mid M(a,{\boldsymbol \ell})=m, {\bf L}(a)={\boldsymbol \ell}, C\right\}dF_{M(a,{\boldsymbol \ell}), {\bf L}(a)\mid C}(m, {\boldsymbol \ell}\mid C)\right]\\
        =& E\left[Y\{a,M(a)\}\right],
    \end{align*}
    and
    \begin{align*}
        &E\left(Y\left[a,G\left\{a^*\mid {\bf C}, {\bf L}(a^*)\right\}\right]\right)\\
        =& E\left(\iiint\limits_{{\boldsymbol \ell^*},{\boldsymbol \ell},m}E\left[Y(a,m)\mid G\{a^*\mid C,{\bf L}(a^*)\}=m,{\bf L}(a^*)={\boldsymbol \ell^*}, {\bf L}(a)={\boldsymbol \ell}, C\right]\right.\\
        &\times \left.dF_{G\{a^*\mid C,{\bf L}(a^*)\}\mid {\bf L}(a^*),{\bf L}(a),C}(m\mid {\boldsymbol \ell^*},{\boldsymbol \ell},C)dF_{{\bf L}(a^*),{\bf L}(a)\mid C}({\boldsymbol \ell^*},{\boldsymbol \ell}\mid C)\right)\\
        =& E\left[\iiint\limits_{{\boldsymbol \ell^*},{\boldsymbol \ell},m}E\left\{Y(a,{\boldsymbol \ell},m)\mid M(a^*,{\boldsymbol \ell^*})=m,{\bf L}(a^*)={\boldsymbol \ell^*}, {\bf L}(a)={\boldsymbol \ell}, C\right\}\right.\\
        &\times \left.dF_{G\{a^*\mid C,{\bf L}(a^*)\}\mid {\bf L}(a^*),C}(m\mid {\boldsymbol \ell^*},C)dF_{{\bf L}(a^*),{\bf L}(a)\mid C}({\boldsymbol \ell^*},{\boldsymbol \ell}\mid C)\right]\\
        =& E\left(\iiint\limits_{{\boldsymbol \ell^*},{\boldsymbol \ell},m}E\left[Y\{a,M(a^*)\}\mid M(a^*,{\boldsymbol \ell^*})=m,{\bf L}(a^*)={\boldsymbol \ell^*}, {\bf L}(a)={\boldsymbol \ell}, C\right]\right.\\
        &\times \left.dF_{M(a^*)\mid {\bf L}(a^*),C}(m\mid {\boldsymbol \ell^*},C)dF_{{\bf L}(a^*),{\bf L}(a)\mid C}({\boldsymbol \ell^*},{\boldsymbol \ell}\mid C)\right)\\
        =& E\left(\iiint\limits_{{\boldsymbol \ell^*},{\boldsymbol \ell},m}E\left[Y\{a,M(a^*)\}\mid M(a^*,{\boldsymbol \ell^*})=m,{\bf L}(a^*)={\boldsymbol \ell^*}, {\bf L}(a)={\boldsymbol \ell}, C\right]\right.\\
        &\times \left.dF_{M(a^*), {\bf L}(a^*),{\bf L}(a),C}(m,{\boldsymbol \ell^*},{\boldsymbol \ell},C)\right)\\
        =& E\left[Y\left\{a,M(a^*)\right\}\right],
    \end{align*}
    hence
    \[E\left(Y\left[a,G\left\{a\mid {\bf C}, {\bf L}(a)\right\}\right]\right)-E\left(Y\left[a,G\left\{a^*\mid {\bf C}, {\bf L}(a^*)\right\}\right]\right)=\text{NIE}.\]
    Since the NIE is not nonparametrically identified in this setting, neither then is \[E\left(Y\left[a,G\left\{a\mid {\bf C}, {\bf L}(a)\right\}\right]\right)-E\left(Y\left[a,G\left\{a^*\mid {\bf C}, {\bf L}(a^*)\right\}\right]\right).\]
\end{proof}

\begin{proof}[Proof of Theorem \ref{thm:no-LM-interaction}]
    We can infer from \cite{tchetgen2014identification} that under the assumption of no mean ${\bf L}$--$M$ interaction on $Y$ on the additive scale, the NIE is nonparametrically identified by
    \[E\left[\int_m E(Y\mid m,{\boldsymbol \ell'},a,C)\left\{dF_{M\mid A,C}(m\mid a,C)-dF_{M\mid A,C}(m\mid a^*,C)\right\}\right]\]
    for any ${\boldsymbol \ell'}$. Under this assumption, we also have
    \begin{align*}
        \text{NIE}^{\text{R}} =& E\left[\int_m E\left\{E(Y\mid m,{\bf L},a,{\bf C})\mid a,{\bf C}\right\}\{dF_{M\mid A,{\bf C}}(m\mid a,{\bf C}) - dF_{M\mid A,{\bf C}}(m\mid a^*,{\bf C})\}\right]\\
        =& E\left[\int_m E\left\{E(Y\mid m',{\bf L},a,{\bf C}) + E(Y\mid m,{\boldsymbol \ell'},a,{\bf C}) - E(Y\mid m',{\boldsymbol \ell'},a,{\bf C})\mid a,{\bf C}\right\}\right.\\
        &\times\left.\{dF_{M\mid A,{\bf C}}(m\mid a,{\bf C}) - dF_{M\mid A,{\bf C}}(m\mid a^*,{\bf C})\}\right]\\
        =& E\left[\int_m E(Y\mid m,{\boldsymbol \ell'},a,C)\left\{dF_{M\mid A,C}(m\mid a,C)-dF_{M\mid A,C}(m\mid a^*,C)\right\}\right]\\
        =& \text{NIE}
    \end{align*}
    for all ${\boldsymbol \ell'}$.
\end{proof}

\begin{proof}[Proof of Theorem \ref{thm:separable-effects}]
    The FFRCISTG corresponding to the DAG in Figure 3 is a submodel of the FFRCISTG corresponding to the DAG in Figure 1, hence the NIE$^{\text{R}}$ is nonparametrically identified by $\Psi^{\text{NIE}}(P)$ as it is in the larger model. Since the NIE is also nonparametrically identified by $\Psi^{\text{NIE}}(P)$ under this submodel, the two are equivalent, and since the NIE satisfies the sharp null criterion, so must the NIE$^{\text{R}}$.
\end{proof}

\begin{proof}[Proof of Theorem \ref{thm:m-always-affects-y}]
    In this case, the sharp mediational null implies that $M(a)=M(a^*)$ must hold almost surely, in which case $G(a)$ and $G(a^*)$ follow the same distribution, and $\text{NIE}^{\text{R}}=0$. Thus, the $\text{NIE}^{\text{R}}$ satisfies the sharp null criterion, and in turn must also satisfy the sharper null criterion.
\end{proof}

\bibliographystyle{apalike}

\bibliography{references}